\documentclass[aps, prx,superscriptaddress,twocolumn, longbibliography]{revtex4-1}%
\usepackage{amsfonts, xcolor, mathtools, amsmath, amssymb, graphicx, mathrsfs, physics, bbm, times, subfigure, wrapfig}
\definecolor{myurlcolor}{rgb}{0,0,0.7}
\definecolor{myurlcolor1}{rgb}{0,0.7,0.1}
\definecolor{myrefcolor}{rgb}{0,0,0.7}
\usepackage[unicode=true,pdfusetitle, bookmarks=true,bookmarksnumbered=false,bookmarksopen=false, breaklinks=false,pdfborder={0 0 0},backref=false,colorlinks=true, linkcolor=myrefcolor,citecolor=myurlcolor1,urlcolor=myurlcolor]{hyperref}
\definecolor{mybox}{rgb}{0.5,0.5,0.3}
\providecommand{\U}[1]{\protect\rule{.1in}{.1in}}
\newtheorem{theorem}{Theorem}

\newtheorem{corollary}{Corollary}
\newtheorem{definition}{Definition}
\newtheorem{example}{Example}

\newtheorem{proposition}{Proposition}
\newtheorem{remark}{Remark}
\newenvironment{proof}[1][Proof]{\noindent\textbf{#1.} }{\ \rule{0.5em}{0.5em}}

\def\supp{\operatorname{supp}}
\def\ex{\operatorname{ext}}
\newcommand{\mc}[1]{\mathcal{#1}}
\newcommand{\msc}[1]{\mathscr{#1}}

\begin{document}

\title{Partial order on passive states and Hoffman majorization in quantum thermodynamics}


\author{Uttam Singh}
\email{utsingh@ulb.ac.be}
\affiliation{Centre  for Quantum Information and Communication, {\'E}cole polytechnique de Bruxelles,\\
CP 165, Universit{\'e} libre de Bruxelles, 1050 Bruxelles, Belgium}
\author{Siddhartha Das}
\email{sidddas@ulb.ac.be}
\affiliation{Centre  for Quantum Information and Communication, {\'E}cole polytechnique de Bruxelles,\\
CP 165, Universit{\'e} libre de Bruxelles, 1050 Bruxelles, Belgium}
\author{Nicolas J. Cerf}
\email{ncerf@ulb.ac.be}
\affiliation{Centre  for Quantum Information and Communication, {\'E}cole polytechnique de Bruxelles,\\
CP 165, Universit{\'e} libre de Bruxelles, 1050 Bruxelles, Belgium}

\begin{abstract}
Passive states, i.e., those states from which no work can be extracted via unitary operations, play an important role in the foundations and applications of quantum thermodynamics. They generalize the familiar Gibbs thermal states, which are the sole passive states being stable under tensor product. Here, we introduce a partial order on the set of passive states that captures the idea of a passive state being \textit{virtually cooler} than another one. This partial order, which we build by defining the notion of  \textit{relative passivity}, offers a fine-grained comparison between passive states based on virtual temperatures (just like thermal states are compared based on their temperatures). We then characterize the quantum operations that are closed on the set of virtually cooler states with respect to some fixed input and output passive states. Viewing the activity, i.e., nonpassivity, of a state as a resource, our main result is then a \textit{necessary and sufficient} condition on the transformation of a class of pure active states under these relative passivity-preserving operations. This condition gives a quantum thermodynamical meaning to the majorization relation on the set of nonincreasing vectors due to Hoffman. The maximum extractable work under relative passivity-preserving operations is then shown to be equal to the ergotropy of these pure active states. Finally, we are able to fully characterize passivity-preserving operations in the simpler case of qubit systems, and hence to derive a state interconversion condition under passivity-preserving qubit operations. The prospect of this work is a general resource-theoretical framework for the extractable work via quantum operations going beyond thermal operations. 
\end{abstract}

\maketitle

\section{Introduction}
A major focus in the area of thermodynamics is work extraction. The laws of thermodynamics, which are expressed as a set of phenomenological rules, govern the transformations of states leading to work extraction (see e.g. \cite{Callen1985}). Traditionally, these laws are given for macroscopic systems in equilibrium at the so-called thermodynamic limit; therefore, they are inadequate to deal with the thermodynamics of systems at nano scale. However, with the recent technological developments and quest for miniaturization \cite{Scovil1959, Geusic1967, Scully2002, Howard1997, Hanggi2009, Faucheux1995, Pusz1978, Lenard78,Alicki1979, Rousselet1994}, the need for developing a microscopic version of thermodynamics, where quantum effects play a crucial role, becomes pertinent. Recently, a great deal of effort has been put to understand and possibly deduce thermodynamical laws that take quantum effects into account, which has led to the research area known as quantum thermodynamics \cite{Michal2013, Brandao2013, Brandao2015b, Cwiklinnski2015, Goold2016,Vinjanampathy2016, Millen2016,  Faist2015, Faist2015A, Nicole2016, Uzdin2018, Koukoulekidis2019, Faist2019, Uzdin2019, Marvian2020, Francica2020}.

Since thermodynamics is, fundamentally speaking, a theory of state transformations, the resource-theoretic framework offers a natural ground for the development of quantum thermodynamics.  For example, in Refs. \cite{Michal2013, Brandao2013}, the Gibbs (thermal) state at a fixed temperature is considered as free and thermal operations are considered as free operations, giving rise to a resource theory of thermal operations. Under these restrictions, the laws of state transformations have been described, and the necessary and sufficient conditions for the transformation of energy-diagonal states under thermal operations have been expressed as a thermo-majorization condition~\cite{Brandao2015b}. These transformations are based on a bath system in a fixed thermal state, which acts as catalyst.

The importance of work extraction cannot be overemphasized in quantum thermodynamics and therefore the characterization of quantum states based on their thermodynamical merits is very important in itself. It is clear from the definition of the passive states that no work can be extracted from them via cyclic unitary operations \footnote{By cyclic unitary operations, we mean unitary operations that are such that the Hamiltonian of the system returns to its initial expression at the end of a cycle.}. However, there are various thermodynamical contexts where one is interested in work extraction and state transformations of a quantum system interacting with its environment, which may be inaccessible to an observer in general and deviate from a thermal bath. A pertinent question is to determine general conditions under which state transformations can take place for the extraction of work from an individual system. In this work, we take a resource-theoretic approach to deal with the thermodynamics of a single quantum system interacting with a nonthermal environment. In this context, it seems natural to consider the operations which cannot create activity (i.e., nonpassivity) starting from a passive state as the useless or free operations that govern the thermodynamics of state transformations.

Given a single quantum system, two different thermal states of the system can be compared and distinguished based on their temperatures. However, such a notion of comparison is absent for passive states, which hinders the treatment of passive-environment maps going beyond thermal maps. As a starting point, we thus introduce here a partial order on the set of passive states that allows for comparison based on their thermodynamic merits. This partial order is a stochastic order that is connected to the concept of ``virtual temperature'' of a quantum system \cite{Brunner2012}. Given a quantum system in state $\rho=\sum_i p_i\op{i}$, the inverse virtual temperatures $\beta_{i,j}$ are defined for each pair of probabilities $(p_i,p_j)$ by the relation $(E_{j}-E_{i})\beta_{i,j}=\ln(p_i/p_j)$, where $E_i$ and $E_j$ are the energy eigenvalues corresponding to the energy eigenstates $\ket{i}$ and $\ket{j}$ of the Hamiltonian of the system. If all $\beta_{i,j}$ are the same, then the state is necessarily a thermal state. The partial order that we define on the set of passive states is naturally expressed in terms of virtual temperatures and therefore, we suggestively refer to it as ``being virtually cooler than". A passive state $\rho$ is said to be virtually cooler than another passive state $\sigma$, denoted as $\rho\succ_{vc}\sigma$, if all virtual temperatures of $\rho$ are lower than those of  $\sigma$. The partial order $\rho \succ_{vc} \sigma$ imposes very stringent conditions on $\rho$ and $\sigma$: it means that $\rho$ is passive \textit{relative to} $\sigma$. In fact,  $\rho \succ_{vc} \sigma$ also implies $\rho\succ \sigma$, where $\succ$ is a preorder called majorization. Moreover, the energy of $\rho$ is always lesser than that of $\sigma$. In another thermodynamic context, we show that if the passive states $\rho$ and $\sigma$ are used for refrigeration of an external qubit system, then a virtually cooler state performs this task better, that is, the use of $\rho$ renders the temperature of the external qubit system lower compared to the case when one uses $\sigma$.

After characterizing the convex set of virtually cooler states than a given passive state, we then consider the class of quantum operations that preserve this partial order. We call such operations as relative passivity-preserving operations and denote them as RPPOs. These operations are necessarily incoherent operation, i.e., operations that map the set of  diagonal states into itself in a fixed reference basis (here, the energy eigenbasis). We then consider the question of state transformations under strictly incoherent RPPOs. In particular, for a specific class of pure active states, we show that the relation we name Hoffman majorization provides a necessary and sufficient condition for state interconversion. Hoffman majorization is a partial order defined on the set of nonincreasing vectors \cite{Hoffman1969}, and we thus uncover here its quantum thermodynamical meaning.

As a consequence of this condition, we also show that the maximal work extraction from a given class of pure active states, which is equal to ergotropy, is achievable with strictly incoherent RPPOs. To elucidate the general constructive proof of our result, we provide the construction of a strictly incoherent RPPO underlying the desired state transformation for a simple qutrit system. Interestingly, for transformations of pure qubit states under strictly incoherent RPPOs, the necessary and sufficient condition based on Hoffman majorization remains true also simply for passivity-preserving operations. The latter operations, which we denote as PPOs, are those that map the set of passive states into itself.

Further, we discuss passivity-preserving operations in general and characterize them completely in the case of qubit systems by explicitly providing their Kraus operators. Considering again work extraction, we show that the ergotropy is simply achievable via passivity-preserving operations for a given class of pure active qubit states. Moreover, as an example of passivity-preserving operations, we introduce what we call activity-breaking operations (ABOs) in analogy with entanglement-breaking operations. ABOs are defined as the operations that map any input state into the set of passive states. We show that any activity-breaking operation can be written as some measure and prepare channel.

Finally, coming back to a resource-theoretical framework for extractable work, we introduce several monotones on the set of active pure states, which are nonincreasing under RPPOs. These should be useful to better understand the nature of the resource consisting in not being ``virtually cooler than'' in a thermodynamical scenario.

The rest of the paper is organized as follows. In Sec. \ref{sec:prelims}, we introduce the classical preliminaries, in particular the concept of Hoffman majorization which induces a partial order on the set of nonincreasing probability vectors. In Sec. \ref{sect:passive-states-PPO}, we turn to quantum thermodynamics and recall the set of passive states (extending the nonincreasing vectors) together with the notion of passivity-preserving quantum operations and incoherent (and strictly incoherent) quantum operations, which are necessary for the rest to follow. We also characterize the passivity preserving -- and, conversely, the activity breaking -- quantum operations in terms of their Kraus decomposition. In Sec. \ref{sec:ver-coo}, we define relative passivity, namely the condition that a state is passive relative to another passive state. This leads us to define the notion of a ``virtually cooler" quantum state and explore its thermodynamical consequences. We then introduce the notion of relative passivity-preserving operations (RPPOs) and illustrate it with an example. In Sec. \ref{sec:stat-tran}, we come back to our central question, namely the thermodynamics of work extraction. We provide the necessary and sufficient condition for the transformation of a specific class of active pure states under RPPOs, which makes the connection with Hoffman majorization and gives it a thermodynamical meaning. We conclude in Sec. \ref{sec:conc} with further discussion on the ramification of the concepts introduced here. In Appendix \ref{append:Hoff-Maj}, we provide more details on the main theorems at the heart of Hoffman majorization and prove another equivalent characterization of Hoffman majorization in terms of sequence of passive $t$-transforms. 
In Appendix~\ref{append:qubits}, we focus on passivity-preserving operations for qubit systems, exploiting the fact that these are both incoherent and strictly incoherent at the same time. This enables us to explicitly characterize the form of qubit passivity-preserving operations in terms of Kraus operators. In Appendix \ref{append:exp}, we construct an explicit RPPO suitable for pure state transformations in the case of a qutrit system, which is instructive to understand RPPOs in general. Then, in Appendix \ref{append:thermo}, we explore the thermodynamical implication of the relative passivity condition, namely the refrigeration by using a passive state that is virtually cooler than another passive state. Finally, in Appendix \ref{append:extractablework}, we briefly touch upon the problem of defining the notion of extractable work under quantum channels beyond the thermal case and provide an example of work extraction under a RPPO for a qubit system.







\section{Classical preliminaries}
\label{sec:prelims}
We start by introducing the notion of majorization over the set of nonincreasing probability vectors, which we name Hoffman majorization. Our main result, which we present in Section~\ref{sec:stat-tran}, relies on an extension of Hoffman majorization to quantum passive states and finds a natural application to quantum thermodynamics (see Theorems \ref{th:qud-iff} and \ref{th:qubit-iff}).

\subsection{Nonincreasing vectors and Hoffman matrices}
\label{sec:prelims-A}
A probability vector $\vec{p}=(p_0,\cdots,p_{d-1})^T$, where $p_i\in\mathbb{R}_+$ and $\sum_{i=0}^{d-1}p_i=1$, is said to be a nonincreasing vector if it is such that $p_i \leq p_j$,  $\forall i >  j$. The set of nonincreasing probability vectors of dimension $d$, denoted as $\mathcal{S}(d)$, is a convex set whose  extreme points are given by $d$ vectors 
\begin{align}
\vec{e}_k = \frac{1}{k+1}(\overbrace{1,\cdots,1}^{k+1},\overbrace{0,\cdots,0}^{d-k-1})^T,
\end{align}
with $k=0,\cdots,d-1$. A simple proof of this statement can be found as Lemma $1$ in Ref.~\cite{Kastner1998}. In the rest of the paper, we denote nonincreasing probability vectors as passive vectors and the set $\mc{S}(d)$ as the set of passive vectors. The reason for this nomenclature is the close connection between the nonincreasing vectors and the passive states (see Section \ref{sect:passive-states-PPO}).

\bigskip
\noindent
{\it Hoffman matrices.--} 
We say that a $d\times d$ matrix $R$ is a Hoffman matrix if, for all $0\leq i,j\leq d-1$, it satisfies the following conditions \cite{Hoffman1969}:
\begin{align*}
&(a)~ R_{0,d-1}\geq 0 \, ;~(b)~\sum_{j=0}^{d-1}R_{i,j}=1\, ;~(c)~R_{i,j}=R_{j,i} \, ;\\
&(d)~R_{i,j}+R_{i-1,j+1}\geq R_{i-1,j}+ R_{i,j+1} ,~\forall~ i \leq j.
\end{align*}
Conditions (a) to (d) imply, in particular, that Hoffman matrices are doubly-stochastic matrices that are symmetric and map the set $\mathcal{S}(d)$ of passive vectors into itself. Let us denote the set of $d\times d$ Hoffman matrices by $\mathcal{R}(d)$.
In the simple case of $d=2$, the set of Hoffman matrices $\mathcal{R}(2)$ reduces to the so-called $t$-transforms with a constraint, that is,
\begin{align}
R=\begin{pmatrix}
t& 1-t\\
1-t & t
\end{pmatrix},  \qquad \mathrm{with~} t\geq 1/2.
\end{align}

\medskip
\noindent
{\it Explicit construction.--} Before exposing the role of Hoffman matrices in relation with majorization, let us provide an explicit construction.
Let $\mathcal{P}(d)$ denote the set of all $2^{d-1}$ partitions of the set $\{0,\cdots,d-1\}$ such that each part of a partition consists of consecutive integers. For example,
\begin{align*}
&\mathcal{P}(1)=\{(0)\};\\
&\mathcal{P}(2)=\{(0,1),(01)\};\\
&\mathcal{P}(3)=\{(0,1,2), (01,2),(0,12),(012)\}.
\end{align*}
Let us label any such partition in $\mathcal{P}(d)$ as $\tau=(\tau_1,\cdots, \tau_k)$, where $\tau_t$ is a specific part and $k$ is the number of parts within partition $\tau$. For example, in the case $d=3$, the partition $\tau=(0,12)$ consists of $k=2$ parts, namely $\tau=(\tau_1,\tau_2)$ with $\tau_1=0$ and $\tau_2=12$. For each partition $\tau \in \mathcal{P}(d)$, we define the symmetric doubly-stochastic matrix $M^{\tau}$ by the rule
\begin{align}
M^{\tau}_{i,j}=\begin{cases}
\frac{1}{|\tau_t|} ~\mathrm{if}~i,j\in\tau_t, \quad t\in [1,k]\\
0~~\mathrm{otherwise}.
\end{cases}
\end{align}
where $\tau_t$ can be any of the $k$ parts within partition $\tau$ and $|\tau_t|$ is the size of part $\tau_t$.
We denote the set of the above $2^{d-1}$ matrices by $\mathcal{M}^{\mathcal{P}(d)}$ since each partition $\tau \in \mathcal{P}(d)$ corresponds to a matrix $M^{\tau} \in \mathcal{M}^{\mathcal{P}(d)}$. A generic element of $\mathcal{M}^{\mathcal{P}(d)}$ is thus
\begin{align*}
M^{\tau}=\oplus_{t=1}^{k}\frac{\mathfrak{I}_{|\tau_t |}}{|\tau_t |},
\end{align*}
where $\mathfrak{I}_{|\tau_t |}$ is a  $|\tau_t | \times |\tau_t |$ matrix with all entries equal to one.  For example, in the case of $d=3$, the four partitions and associated matrices are given by
\begin{align*}
&\tau=(0,1,2)
&M^{\tau}=\begin{pmatrix}
1&0 & 0\\
0 & 1 &0\\
0 & 0&1
\end{pmatrix};\\
&\tau=(01,2)
&M^{\tau}=\frac{1}{2}\begin{pmatrix}
1&1& 0\\
1& 1& 0\\
0 & 0 &2
\end{pmatrix};\\
&\tau=(0,12)
&M^{\tau}=\frac{1}{2}\begin{pmatrix}
2& 0 & 0\\
0 & 1 &1\\
0 & 1&1
\end{pmatrix};\\
&\tau=(012)
&M^{\tau}=\frac{1}{3}\begin{pmatrix}
1&1 & 1\\
1& 1& 1\\
1 & 1 &1
\end{pmatrix}.
\end{align*}
The following theorem, due to Hoffman \cite{Hoffman1969}, gives a meaning to this construction.
\begin{theorem}[\cite{Hoffman1969}]
\label{th:hoff1}
The set $\mc{R}(d)$ of $d\times d$ Hoffman matrices coincides with the convex hull of $\mathcal{M}^{\mathcal{P}(d)}$.
\end{theorem}
In other words, any Hoffman matrix $R$ can be constructed as a convex mixture of individual $M^\tau$ matrices. The action of an individual matrix $M^{\tau}$ on a passive vector $\vec{p}$ yields another passive vector $\vec{p}\, '$ such that
\begin{align*}
\vec{p}\, ' = M^{\tau} \, \vec{p} =\oplus_{t=1}^{k} \, {\mathfrak{p}}_t ,
\end{align*}
where
\begin{align*}
\mathfrak{p}_t=  \frac{\sum_{i \in \tau_t} p_i }{|\tau_t |} \,  (\, \overbrace{1,\cdots,1}^{|\tau_t |} \,)^T .
\end{align*}
Since each matrix $M^{\tau} \in \mathcal{M}^{\mathcal{P}(d)}$ maps the set  $\mathcal{S}(d)$ into itself, the same is true for any Hoffman matrix $R\in 
\mc{R}(d)$. The proof of Theorem \ref{th:hoff1} is explained in Appendix \ref{append:Hoff-Maj} for completeness.

\subsection{Hoffman majorization}
Before defining Hoffman majorization,  let us first recall the regular majorization relation between two vectors.


\bigskip
\noindent
{\it Majorization.--} For $x,y\in \mathbb{R}_+^d$, let $x^{\downarrow}$ and $y^{\downarrow}$ be the corresponding vectors with components arranged in nonincreasing order. Then, we say that $x$ is majorized by $y$, denoted by $x\prec y$, if and only if
\begin{align}
\label{eq:majo-cond1}
&\sum_{i=0}^{k}x^{\downarrow}_i \leq \sum_{i=0}^{k} y^{\downarrow}_i , ~~~\forall k=0,\cdots, d-2,~~~\mathrm{and}\\
&\sum_{i=0}^{d-1}x^{\downarrow}_i = \sum_{i=0}^{d-1} y^{\downarrow}_i .
\label{eq:majo-cond2}
\end{align}
Note that Eq. (\ref{eq:majo-cond2}) is automatically verified when $x$ and $y$ are probability vectors, which is the case of interest here.
The relation $x\prec y$ is equivalent to relation $y\succ x$, meaning that $y$ majorizes $x$. A key result of majorization theory is that $x\prec y$ if and only if there exists a doubly-stochastic matrix $D$ such that $x=D y$.
The restriction of the above vectors $x$ and $y$ to the set $\mathcal{S}(d)$ of passive vectors gives rise to what we call Hoffman majorization. This is the content of the following theorem.

\begin{theorem}[\cite{Hoffman1969}]
\label{th:hoff2}
For passive vectors $x,y\in \mathcal{S}(d)$, $x$~is Hoffman-majorized by $y$, denoted as  $x\prec_h y$ [i.e., $x$ and $y$ satisfy Eqs.~\eqref{eq:majo-cond1} and \eqref{eq:majo-cond2} without the need for prior rearrangement], if and only if there exists a Hoffman matrix $R\in \mc{R}(d)$ such that $x=R\, y$.
\end{theorem}

By definition, for $x,y\in\mathcal{S}(d)$, the relation $x \prec_h y$ implies the relation $x \prec y$ and vice-versa (they are equivalent relations over the set of passive vectors) \cite{Marshall2011}. However, the relation $\prec$ is invariant under permutations while the relation $\prec_h$ is not, which is why the Hoffman matrix $R$ involved in relation $\prec_h$ obeys additional constraints beyond being doubly-stochastic. Furthermore, we note that Hoffman majorization $\prec_h$ is a partial order relation as it satisfies the antisymmetry property \footnote{Let $x,y\in\mathcal{S}(d)$. Then $x\prec_h y$ and $y\prec_h x$ imply that $\sum_{i=0}^k x_i = \sum_{i=0}^k y_i$ for all $i=0,\cdots,d-1$, which in turn implies that $x=y$.}, while majorization is only a preorder which does not satisfy the antisymmetry property.

\begin{remark}
The existence of a Hoffman matrix $R$ for relation $x\prec_h y$ to hold is a nontrivial consequence of the restriction to passive vectors $x$ and $y$. 
The fact that the  doubly-stochastic matrix $R$ is symmetric may come as a surprise, but it can be made intuitive by appealing to the geometry of the set $\mc{C}_y$ of vectors $x\in \mc{S}(d)$ satisfying $x\prec_h y$ for a fixed $y\in \mc{S}(d)$. It was shown in Ref. \cite{Hoffman1969} that the set $\mc{C}_y$ is actually a polyhedron and that the extremal points of this polyhedron are of the form $M^{\tau} y$, where $M^{\tau}\in \mc{M}^{\mc{P}(d)}$. In fact, this is the main idea behind one of the proofs of Theorem \ref{th:hoff2} presented in Ref. \cite{Hoffman1969} (see Appendix \ref{append:Hoff-Maj} for a simple proof of Theorem \ref{th:hoff2}).
\end{remark}

\begin{remark}
According to Theorem \ref{th:hoff2}, for two passive vectors $\vec{p}$ and $\vec{q}$, we have $\vec{p}\prec_h\vec{q}$ if and only if there exists a Hoffman matrix $R$ such that $\vec{p}=R\vec{q}$. As mentioned earlier, for two dimensional case ($d=2$), every Hoffman matrix is a $t$-transform with constraint, i.e.,
\begin{align}
R=\begin{pmatrix}
t & \bar{t}\\
\bar{t} & t
\end{pmatrix},
\end{align}
where $\bar{t}=1-t$ and $1/2\leq t\leq 1$. Thus, the condition for Hoffman majorization for two-dimensional vectors simply boils down to the existence of such a $t$-transform, which we call a passive $t$-transform. The extension to higher dimensions  ($d>2$) is discussed in Appendix \ref{append:Hoff-Maj}, where we connect Hoffman majorization to the existence of a decomposition into passive $t$-transforms, see Theorem \ref{th:exist-t-for-hm}. This also leads us to the consider the existence of a doubly-stochastic matrix connecting $\vec{p}$ and $\vec{q}$ that is not necessarily symmetric (hence, it is not a Hoffman matrix $R$) and maps the set $\mc{S}(d)$ into itself.
\end{remark}



\section{Passive states and passivity-preserving quantum operations}
\label{sect:passive-states-PPO}
Let us move to the quantum scenario and introduce the set of passive states (extending the passive vectors) together with the notion of passivity-preserving quantum operations. We will also discuss active states (i.e., nonpassive states) as well as what we call activity-breaking quantum operations.

\subsection{Passive states}
Consider a separable Hilbert space $\mc{H}_S$ associated with a finite-dimensional quantum system $S$ with Hamiltonian $\hat{H}_S$. Let $\{\ket{i}\}$ with $0\leq i\leq d-1$ be the energy eigenbasis, which forms an orthonormal basis of $\mc{H}_S$, and let us denote as $E_i$ the energy eigenvalue corresponding to $\ket{i}$. We have
\begin{equation}
\label{eq:ham}
\hat{H}_S=\sum_{i=0}^{d-1}E_i \op{i},
\end{equation} 
where by convention $E_{i}\leq E_{j}$  if  $i \leq j$. A state $\rho_S$ of the system $S$ is passive if and only if its average energy cannot be lowered by unitary operations \cite{Pusz1978, Lenard78}, i.e.,
\begin{equation}
\Tr\{\hat{H}_S \, \rho_S\} \leq \Tr\{\hat{H}_S \, \mathcal{U}_S[\rho_S] \}
\end{equation}
for all unitary operations $\mathcal{U}_S[\cdot]\coloneqq U_S(\cdot)U_S^\dag$, where $U_S$ is a unitary operator. Further, it was shown in Refs. \cite{Pusz1978, Lenard78} that a state $\rho_S$ of the system $S$ with Hamiltonian $\hat{H}_S$ is passive if and only if it can be expressed in the following form 
\begin{equation}
\rho_S=\sum_{i=0}^{d-1}p^{\downarrow}_i \op{i},
\end{equation}
where $\{p^{\downarrow}_i\} \in \mathcal{S}(d)$ denotes a nonincreasing probability distribution, i.e., $p^{\downarrow}_{i} \leq p^{\downarrow}_{j}$, $\forall i>j$. It is evident that any passive state $\rho_S$ commutes with the Hamiltonian $\hat{H}_S$, i.e., $[\rho_S,\hat{H}_S]=0$, hence all passive states are incoherent states if one fixes the energy eigenbasis as the reference basis. (Incoherent states are those states that admit a diagonal density operator in a fixed reference basis). Moreover, the set of passive states is a convex set and the $d$ extremal points of this set are given by $\{\sigma_k\}_{k=0}^{d-1}$, where
\begin{equation}
\sigma_k=\frac{1}{k+1}\sum_{i=0}^{k}\op{i}.
\end{equation}
Let us denote respectively by $\msc{D}(S)$ and $\msc{P}(S)$ the set of all states and the set of all passive states of system $S$. The concept of passive states, which was introduced to derive and justify statistical physics starting from certain physical assumptions on quantum states and operations \cite{Pusz1978, Lenard78}, has played a major role in the development of quantum thermodynamics and has been applied to several contexts \cite{Allahverdyan2004, Skrzypczyk2015, Gemmer2015, Korzekwa2016, Brown2016, Perarnau2015, Perarnau-Llobet2015}. In particular, passive states can be viewed as a natural generalization of thermal states $\rho(\beta)=Z^{-1}e^{-\beta \hat{H}_s}$, with $Z=\sum_{k=0}^{d-1}e^{-\beta E_k}$ , where the negative exponential ensures the nonincreasing probability distribution. The set of all thermal states of system $S$ (for all $\beta\ge 0$) will be denoted as $\msc{B}(S)$ in the following. 

In the context of quantum thermodynamics, a major theme is to obtain the maximal amount of work that can be extracted using some quantum operations (see, e.g., Refs. \cite{Allahverdyan2004, Michal2013, Faist2019}).
In particular, the maximal extractable work from a system in state $\rho$ under unitary transformations, denoted as $\mc{U}$, is given by
\begin{align*}
W_{\max}^\mc{U}(\rho)\coloneqq \max_{\mc{U}_S}\Tr\{\hat{H}_S\left(\rho-\mc{U}_S[\rho]\right)\}.
\end{align*}
The above expression is also called the ergotropy of the state $\rho$ \cite{Allahverdyan2004} and is of course zero for a passive state. For an arbitrary state $\rho=\sum_{i=0}^{d-1}q_i \op{q_i} $, with $q_{i}\leq q_{j}$, $\forall i > j$, the ergotropy is thus given by \cite{Allahverdyan2004}
\begin{align}
\label{def:erg}
W_{\max}^\mc{U}(\rho) = \sum_{i,k=0}^{d-1} q_i E_k \left(|\bra{q_i}\ket{k}|^2-\delta_{i,k}\right).
\end{align}
Note that unitary transformations leave the entropy of the system unchanged, thereby making the change in energy of the system a reasonable quantifier for the extractable work. However, if a quantum operation changes the entropy of the system, the  difference in energy before and after the operation does not remain a valid quantifier for the extractable work (e.g., see Appendix \ref{append:extractablework} for a discussion).

Following the concept of passive states, which are useless for work extraction under unitary operations, it is natural to consider quantum operations that are useless for work extraction. In particular, we now consider the class of quantum operations that preserve the set $\msc{P}(S)$ of passive states.

\subsection{Passivity-preserving operations (PPO)}

In general, a quantum operation or quantum channel $\msc{N}_{S\to S'}: \msc{D}(S)\to \msc{D}(S')$ is a completely positive, trace-preserving map that acts on the input system $S$ and yields the output system $S'$. A quantum channel $\mc{N}_{S\to S'}$ is said to be passivity-preserving if its output is passive whenever its input is passive, i.e., $\mc{N}(\rho_S)\in\msc{P}(S'), \forall \rho_{S}\in\msc{P}(S)$. Thus, since passivity-preserving operations (PPO) cannot create an active (i.e., nonpassive) state starting from a passive state, they cannot be used to create ergotropy :  if $\mc{N}_{S\to S'}$ is a PPO and $\rho_{S}\in\msc{P}(S)$ then $W_{\max}^\mc{U}\left(\mc{N}_{S\to S'}(\rho_S)\right)=0$.

As already mentioned, passive states are incoherent states with respect to the energy eigenbasis. As a result, the concept of passivity-preserving operations is related with the concept of incoherent-preserving operations (i.e., channels that map incoherent states to incoherent states). From now on, we will refer to incoherent-preserving operations simply as \textit{incoherent operations} for brevity ~\footnote{These operations have sometimes been referred to as maximally incoherent operations, see, for example, Refs. \cite{Aberg2006, Chitambar2016}, but we find it more appropriate to call them incoherent-preserving operations or simply incoherent operations.}. 
Further, we will define as \textit{strictly incoherent operations} the incoherent operations that admit a strictly incoherent Kraus decomposition (i.e., such that all individual Kraus operators also map incoherent states to incoherent states \footnote{These strictly incoherent operations have sometimes been called incoherent operations, see for example Refs. \cite{Baumgratz2014}.}).
For more details on the resource theory of coherence see, e.g., Refs. ~\cite{Baumgratz2014, Winter2016, Marvian2013, Streltsov2015, Streltsov2017A, Lami2020}.

The following proposition is a structural statement about the passivity-preserving operations.

\begin{proposition} All passivity-preserving operations are incoherent operations.
\end{proposition}
\begin{proof}
Let us consider that $\mc{N}_{S\to S'}$ is a passivity-preserving operation. Then, for all extremal passive states $\sigma_{k}$,  the output $\gamma_k = \mc{N}_{S\to S'}(\sigma_{k})$ is a passive state for  $0\le k \le d-1$. In particular, 
\begin{align*}
(k+1)\gamma_k - k \gamma_{k-1} &= \mc{N}_{S\to S'}\left((k+1)\sigma_{k}- k\sigma_{k-1} \right)\\
&=\mc{N}_{S\to S'}\left( \op{k} \right).
\end{align*}
Since, $\mc{N}_{S\to S'}$ is a quantum channel,  $\mc{N}_{S\to S'}\left( \op{k} \right) \in \msc{D}(S')$, hence $[(k+1)\gamma_k - k \gamma_{k-1} ] \in \msc{D}(S')$. Further, $\gamma_k$ and $\gamma_{k-1}$ are diagonal states in energy basis, which implies that $\mc{N}_{S\to S'}\left( \op{k} \right)$ is diagonal in energy basis for all $k$. Thus we conclude that all passivity-preserving operations map diagonal states to diagonal states in energy basis.
\end{proof}

\begin{remark}
The converse of the above proposition does not hold. For example, a permutation of energy eigenstates is an example of incoherent operation, however, it is not a passivity-preserving operation.
\end{remark}

\bigskip
\noindent
Passivity-preserving operations exhibit the following useful properties. 

\bigskip
\noindent
{\it (P1) Convexity:} The set of passivity-preserving channels is convex in the sense that if $\mc{N}_{S\to S'}$ and $\mc{N}'_{S\to S'}$ are two passivity-preserving channels, then  $p \mc{N}_{S\to S'}+(1-p)\mc{N}'_{S\to S'} $ with $0\le p\le 1$ is also a passivity-preserving channel.

\medskip
\noindent
{\it (P2) Composability:} The composition of passivity-preserving channels is again a passivity-preserving channel, that is, if $\mc{N}_{S_1\to S_2}$ and $\mc{N}'_{S_2\to S_3}$ are two passivity-preserving channels, then $\mc{N}_{S_2\to S_3}\circ \mc{N}'_{S_1\to S_2}$ is also a passivity-preserving channel.

\medskip
\begin{remark}
In the special case of two-dimensional (qubit) input and output systems $S$ and $S'$, the set of incoherent operations and the set of strictly incoherent operations are known to coincide \cite{Chitambar2016}. Hence, all qubit passivity-preserving operations are not only incoherent operations but also strictly incoherent operations (all Kraus operators are incoherent).
\end{remark}

This property can be used to provide an explicit characterization of qubit passivity-preserving operations in terms of five (incoherent) Kraus operators, as shown in Appendix \ref{append:qubits}. It is also noted in this Appendix that qubit passivity-preserving operations enjoy another nice physical property, in terms of a Stinespring dilation comprising an energy-preserving unitary operation and a passive environment.


Now consider the case where the input and output systems $S$ and $S'$ have the same (arbitrary) dimension. By restricting to strictly incoherent passivity-preserving operations from pure states to pure states (with some restriction on the states), it appears that such operations can be explicitly characterized in terms of incoherent Kraus operators as follows.


\begin{proposition}
If a strictly incoherent passivity-preserving operation is used to transform a pure state into another pure state (with the restriction that each energy level in both states has nonzero amplitude and both states have the same dimension), then it is sufficient to consider Kraus operators that contain one and only one nonzero element in each row and column.
\end{proposition}
\begin{proof}
If an operation $\Phi:=\{K_\mu\}$ transforms a pure state $\ket{\psi}$ into another pure state $\ket{\phi}$ such that both states have nonzero amplitudes for each energy level, i.e., $\ket{\psi}=\sum_{i=0}^{d-1} \alpha_i \ket{i}$  and $\ket{\phi}=\sum_{i=0}^{d-1} \beta_i \ket{i}$ with $\alpha_i,\beta_i>0$, $\forall i$, and $\sum_{i=0}^{d-1}|\alpha_i|^2 =\sum_{i=0}^{d-1}|\beta_i|^2=1$, then, $\Phi(\op{\psi})=\op{\phi}$ implies that the Kraus operators $\{K_{\mu}\}$ satisfy
\begin{align}
\label{eq:pure-kraus}
K_{\mu}\ket{\psi}=c_\mu \ket{\phi}, ~\forall \mu.
\end{align}
Here, $c_\mu$ are proportionality constants such that $\sum_\mu|c_\mu|^2=1$. Since each component of $\ket{\phi}$ in energy basis is nonzero, any $K_\mu$ which has a zero row cannot contribute to the state transformation owing to Eq. \eqref{eq:pure-kraus}. Therefore  every Kraus operator that contributes to this pure state transformation must have all its rows containing at least one nonzero entry. Also, since we consider a strictly incoherent operation, each Kraus operator can have at most one nonzero entry in each column. These two constrains together imply that each Kraus operator must have one and only one nonzero entry in each row and column. 
\end{proof}

\bigskip
Finally, to complete this section, let us define active states (in particular, a class of pure active states which will play a central role in our main theorem in Sec.~\ref{sec:stat-tran}) as well as the notion of activity-breaking operations.

\subsection{Active states}
A state $\rho \in \msc{D}(S)$ is said to be an active state if it is not passive. By definition, unitary operations may decrease the energy of any active state and, therefore, active states can be thought of as a source of usable energy, e.g., a charged quantum battery. All pure states except the ground state of $S$ are, by definition, active states. Among all pure states, there is a particular set $\mathfrak{D}$ of pure states which we will use, namely, the set of states of the form
\begin{align*}
\ket{\psi}=\sum_{i=0}^{d-1}e^{-i\, \theta_i}\sqrt{p_i}\ket{i},
\end{align*}
where $\{ p_{i} \} \in \mathcal{S}(d)$, that is $p_{i} \leq p_{j}$, $\forall i> j$, and $p_0\in(0,1)$. For example, a pure state of the form
\begin{align*}
\ket{\psi}=\sum_{i=0}^{d-1}\sqrt{\frac{e^{-\beta E_i}}{Z}}\ket{i}
\end{align*}
lies in the set $\mathfrak{D}$, where $Z=\mathrm{Tr}[e^{-\beta \hat{H}_S}]$ and $\beta\geq 0$. In the simplest case of qubits, the pure states in the set $\mathfrak{D}$ lie on the surface of upper half of the Bloch sphere. Thus, although it may seem very constrained, the set $\mathfrak{D}$ contains a significant part of all nontrivial pure quantum states.

\subsection{Activity-breaking operations (ABO)}

We define a special class of passivity-preserving operations that are interesting from a resource-theoretic viewpoint, namely the activity-breaking operations.  Naturally, a quantum channel $\mc{N}_{S\to S'}$ is called \textit{activity-breaking} if the output state of the channel is always passive for any input state, i.e., $\mc{N}(\rho_S)\in\msc{P}(S')$ for all $\rho_S\in\msc{D}(S)$. This is a straightforward analog to the notions of entanglement-breaking or coherence-breaking channels.
The following theorem gives a complete characterization of activity-breaking channels.
\begin{theorem}
A quantum channel $\mc{N}_{S\to S'}$, where the Hamiltonian of system $S'$ is denoted as $H_{S'}=\sum_{k=0}^{d-1}E'_k\op{E'_k}{E'_k}$  and $\dim(\mc{H}_{S'})=\dim(\mc{H}_{S})$, is activity-breaking if and only if it admits the following form
\begin{equation}
\mc{N}(\rho_{S})=\sum_{k=0}^{d-1}\Tr\{\rho_{S}\Gamma_{k}\}\op{E'_k}_{S'},
\label{eq:ABO}
\end{equation}
where the operators $\{\Gamma_k\}_{k=0}^{d-1}$ form a positive-operator-valued measure (POVM), i.e., $\sum_{k=0}^{d-1}\Gamma_{k}=\mathbbm{1}_{S}$ and $\Gamma_{k}\geq 0$, $\forall k$, satisfying $\Gamma_{k} \leq \Gamma_{k'}~\forall k>k'$.
\end{theorem}
\begin{proof}
To prove the ``if" statement of the above theorem, we employ the fact that activity-breaking channels have to be coherence-breaking \cite{Bu2016}. Let us assume that the channel $\mc{N}_{S\to S'}$ is activity-breaking. Since it is then a special case of a coherence-breaking channel, the channel $\mc{N}_{S\to S'}$ can be expressed as \cite[Theorem~2]{Bu2016}
\begin{equation}
\mc{N}(\rho_{S})=\sum_{k=0}^{d-1}\Tr\{\rho_{S}\Gamma_k\}\op{E'_k}_{S'},
\end{equation}
such that $\sum_{k}\Gamma_k=\mathbbm{1}_{S}$ and $\Gamma_k\geq 0$ for all $k\in\{0,\cdots,d-1\}$. Additionally, imposing that $\mc{N}(\rho_{S})\in\msc{P}(S')$ since the channel is activity-breaking implies that for all $k>k'$
\begin{equation}
\Tr[\rho_{S}\Gamma_k]\leq \Tr[\rho_{S}\Gamma_{k'}].
\end{equation}
This inequality holds for all input states $\rho_{S}$ only if 
$\Gamma_{k'}-\Gamma_k \geq 0$. This condition necessarily requires $\supp(\Gamma_{k})\subseteq \supp(\Gamma_{k'})$ as $\Gamma_{k}\geq 0$.

To prove the converse, we assume that $\mc{N}(\rho_{S})$ is given by Eq. (\ref{eq:ABO}) and notice that $\sum_{k=0}^{d-1}\Tr\{\Gamma_{k}\rho_{S}\}\op{E'_k}_{S'}\in\msc{P}(S')$ regardless of $\rho_{S}\in\msc{D}(S)$ as soon as the POVM $\{\Gamma_k\}$ satisfies $\Gamma_{k}\leq \Gamma_{k'}$, $\forall k>k'$. This concludes the proof of the theorem.
\end{proof}

As a direct consequence of the above theorem, we have following corollary.
\begin{corollary}
An athermality-breaking channel $\mc{N}$, i.e., a channel which outputs a thermal state for any input state, has the following form:
\begin{equation}
\mc{N}(\rho_{S})=\sum_{k=0}^{d-1}\Tr\{\Gamma_{k}\rho_{S}\}\op{E'_k}_{S'},
\end{equation}
where $\{\Gamma_k\}_{k}$ is a POVM such that $
\forall 0\leq k\leq d-1:\ \Gamma_{k}=\frac{e^{-\beta E'_k}}{Z}\mathbbm{1}_{S}$ with $Z=\sum_{k=0}^{d-1}e^{-\beta E'_k}$. 
\end{corollary}


\section{Relative passive states and relative passivity-preserving quantum operations}
\label{sec:ver-coo}

In order to prepare the grounds for our main result in Sec.~\ref{sec:stat-tran}, let us now introduce a partial order relation between passive states, as well as the notion of quantum operations that preserves it. This partial order, which we call \textit{relative passivity}, provides us with a way of comparing two passive states which is analogous to the comparison between thermal states in terms of temperature. Consider again a quantum system $S$ with a Hamiltonian given by Eq.~\eqref{eq:ham} and consider the set of thermal states, namely $\msc{B}(S):=\{\rho(\beta):=Z_\beta^{-1}e^{-\beta \hat{H}_s}\}_{\beta\geq 0}$. The set $\msc{B}(S)$ is endowed with a natural order : for two states $\rho(\beta)$ and $\rho(\beta')$ in $\msc{B}(S)$, $\rho(\beta)$ is said to be cooler than $\rho(\beta')$ if $\beta\geq \beta'$. [Rigorously, we should say that $\rho(\beta)$ is not hotter than $\rho(\beta')$.] Extending on this, we can define a new partial order on the set of passive states $\msc{P}(S)$.


\subsection{Relative passive states}
\begin{definition}
A state $\rho$ is said to be passive relative to some passive state $\sigma$ if and only if $\left(\sigma^{-1/2}\rho \, \sigma^{-1/2}\right)/\mathrm{Tr}[\rho \, \sigma^{-1}]$ is a passive state. 
\end{definition}
The definition above has the following consequence. 

\begin{proposition}
Any state $\rho$ that is passive relative to some passive state $\sigma$, is itself necessarily a passive state.
\end{proposition}

\begin{proof}
Since $\sigma$ is a passive state, we can write $\sigma=\sum_{i=0}^{d-1} p_i \op{i}$ with $p_i\geq p_{i+1}$ for all $0\leq i\leq d-2$.  Let $\rho=\sum_{i,j=0}^{d-1}\rho_{ij}\ket{i}\bra{j}$ be an arbitrary state. Then,
\begin{align}
\sigma^{-1/2}\rho \, \sigma^{-1/2}=\sum_{i,j=0}^{d-1}p_i^{-1/2}\rho_{ij} ~ p_j^{-1/2}\ket{i}\!\bra{j}.
\end{align}
Since $\left(\sigma^{-1/2}\rho \, \sigma^{-1/2}\right)/\mathrm{Tr}[\rho \, \sigma^{-1}]$ is required to be passive, we have $\rho_{ij}=r_i\delta_{ij}$ and
$r_{i} / p_{i}  \geq r_{i+1} / p_{i+1}$, for all $0\leq i \leq d-2$.
This implies $r_{i}\geq r_{i+1}$ as $p_{i}\geq p_{i+1}$ for all $0\leq i \leq d-2$; therefore, $\rho$ is a passive state.
\end{proof}

From the above proposition, it is clear that if a passive state $\rho \! \coloneqq \! \rho(\vec{r})$ is passive relative to another passive state $\sigma \! \coloneqq \! \sigma(\vec{p})$, then we have the condition 
\begin{align}
r_i/r_j\geq p_i/p_j \geq 1 , \qquad \forall i < j.
\end{align}
Intuitively, $\rho(\vec{r})$ is ``more passive'' than $\sigma(\vec{p})$ in the sense that the components of $\vec{r}$ decay faster than those of $\vec{p}$ (see Fig. \ref{fig:virtually-cooler}). Here, we use the convention that $a/0=\infty$ whenever $a\neq 0$. We note that if some $p_i$ is zero, then we have $p_j=0$ for all $j > i$ since the state $\sigma$ is passive. This implies that the corresponding $r_i$'s must vanish as well for $\rho$ to be a passive state relative to $\sigma$.

\begin{figure}[htbp!]
\centering
\includegraphics[width=\columnwidth]{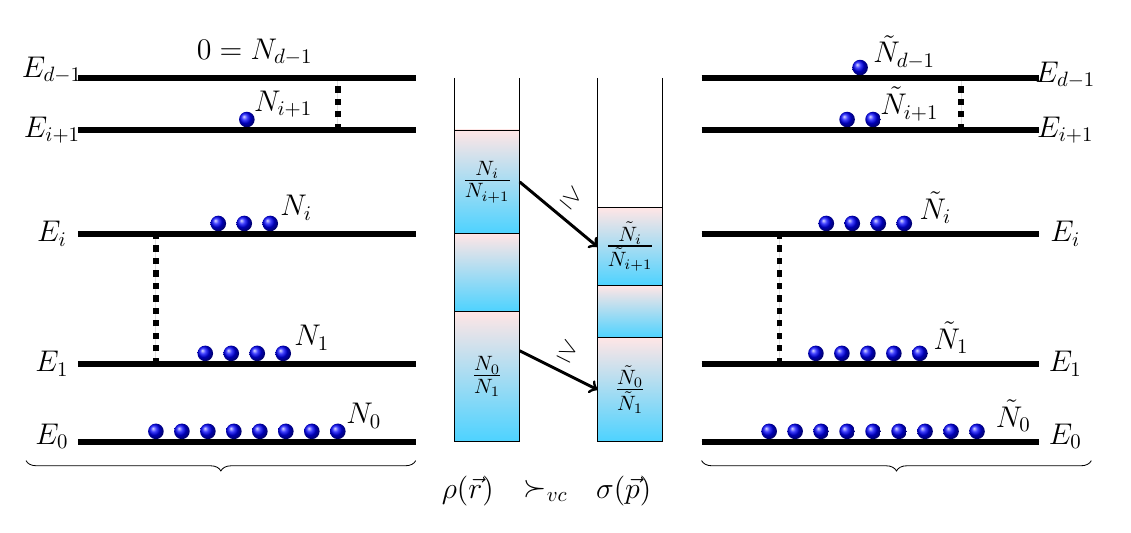}
\caption{Schematic of the notion of virtually cooler passive states. In the schematic, $\{E_i\}_{i=0}^{d-1}$ is the set of energy eigenvalues of the Hamiltonian of the system. $N_i$ and $\tilde{N}_i$ are populations in energy eigenstates $\ket{i}$ corresponding to the passive states $\rho(\vec{r})$ and $\sigma(\vec{p})$, respectively. Now, if $\frac{r_i}{r_{i+1}}:=\frac{N_i}{N_{i+1}}\geq \frac{\tilde{N}_i}{\tilde{N}_{i+1}}:=\frac{p_i}{p_{i+1}}$ for all $i=0,\cdots, d-2$, we say that $\rho(\vec{r})$ is a virtually cooler state than $\sigma(\vec{p})$, i.e., $\rho(\vec{r})\succ_{vc}\sigma(\vec{p})$.}
\label{fig:virtually-cooler}
\end{figure}

Furthermore, the notion of relative passivity on the set of passive states is in close connection with the notion of ``being cooler than'' on the set of thermal states (see Fig. \ref{fig:virtually-cooler}). To see this, let us use the concept of virtual temperatures \cite{Brunner2012, Skrzypczyk2015, Sparaciari2017}. Consider a passive state $\rho(\vec{r})=\sum_{i=0}^{d-1}r_i\op{i} \in \msc{P}(S)$ with $r_i\geq r_{i+1}$ for all $0\leq i\leq d-2$. We can define $d\choose 2$ virtual (inverse) temperatures $\beta_{i,j}$ for all pairs of probabilities appearing in $\rho(\vec{r})$ as follows,
\begin{align}
\beta_{i,j}:=(E_j-E_i)^{-1} \ln \left(\frac{r_i}{r_{j}}\right)  , \qquad \forall i <  j,
\end{align}
where $\beta_{i,j}\geq 0$ as the state $\rho(\vec{r})$ is passive. Similarly, for another passive state $\sigma(\vec{p})=\sum_{i=0}^{d-1}p_i\op{i} \in \msc{P}(S)$ with $p_i\geq p_{i+1}$ for all $0\leq i\leq d-2$, we define the virtual (inverse) temperatures $\beta'_{i,j}$ as
\begin{align}
\beta'_{i,j}:=(E_j-E_i)^{-1} \ln \left(\frac{p_i}{p_{j}}\right)  , \qquad  \forall i <  j.
\end{align}
where $\beta'_{i,j}\geq 0$. Now, expressing the condition that $\rho(\vec{r})$ is passive relative to $\sigma(\vec{p})$ is equivalent to
\begin{align}
\label{eq:virt-cool}
\beta_{i,j} \geq \beta'_{i,j} \geq 0, \qquad \forall i <  j.
\end{align}
We can interpret this condition by saying that $\rho(\vec{r})$ is {\it virtually cooler} than $\sigma(\vec{p})$, in the sense that 
all $d\choose 2$ virtual temperatures $\beta_{i,j}^{-1}$ of $\rho(\vec{r})$ are lower than those of  $\sigma(\vec{p})$ (see Fig. \ref{fig:virtually-cooler}). Thus the partial order relation induced by relative passivity  expresses the physical condition of ``being virtually cooler than'' on the set of passive states, and we denote it by 
\begin{equation}
\rho(\vec{r})\succ_{vc}\sigma(\vec{p}).
\end{equation}
In the special case of thermal states, all virtual temperatures coincide and the condition $\rho(\vec{r})\succ_{vc}\sigma(\vec{p})$ boils down to the condition that $\rho(\vec{r})$ is cooler than $\sigma(\vec{p})$. It is easy to see that the relation $\succ_{vc}$ is a partial order, i.e., (1) $\rho(\vec{r})\succ_{vc}\rho(\vec{r})$ (reflexivity). (2) $\rho(\vec{r})\succ_{vc}\sigma(\vec{p})$ and $\sigma(\vec{p})\succ_{vc}\eta(\vec{s})$ imply $\rho(\vec{r})\succ_{vc}\eta(\vec{s})$ (transitivity). (3) If $\rho(\vec{r})\succ_{vc}\sigma(\vec{p})$ and $\sigma(\vec{p})\succ_{vc}\rho(\vec{r})$, then $\rho(\vec{r})=\sigma(\vec{p})$ (antisymmetry). We also note that the relation $\succ_{vc}$ enables a comparison between passive states but is inadequate to compare a passive state with some nonpassive state.  As a side remark, let us mention that the relation $\succ_{vc}$ appears in mathematical statistics under the name of ``likelihood ratio order'' and has numerous applications including the field of statistical inference, economy and optimal scheduling problems \cite{Shaked2007}. 

We note that the notion of being virtually cooler $\succ_{vc}$ can be connected to Hoffman majorization $\succ_h$, which allows us  in particular to compare the energy of the two states.


\begin{proposition}
\label{prop:vc-implies-Hoffman}
Consider any two passive states $\rho(\vec{r})$ and $\sigma(\vec{p})$. If $\rho(\vec{r})\succ_{vc}\sigma(\vec{p})$, then $\vec{r}\succ_h \vec{p}$ and $E(\rho(\vec{r}))\leq E(\sigma(\vec{p}))$, where $E(\rho)\coloneqq \mathrm{Tr}[\rho\hat{H}_S]$ denotes energy of the state $\rho$.
\end{proposition}
\begin{proof}
From the partial order $\rho(\vec{r})\succ_{vc}\sigma(\vec{p})$, if we choose $k$ such that $0\leq k \leq d-1$, we have
\begin{align*}
\frac{\sum_{i=0}^k r_i}{r_l} \geq \frac{\sum_{i=0}^k p_i}{p_l}, \qquad \forall ~(k+1)\leq l \leq (d-1).
\end{align*}
Inverting the above inequality and summing over $l$, we get
\begin{align*}
\frac{\sum_{l=k+1}^{d-1}r_l}{\sum_{i=0}^k r_i} \leq \frac{\sum_{l=k+1}^{d-1}p_l}{\sum_{i=0}^k p_i}.
\end{align*}
By adding one on both sides, using that $\sum_{i=0}^{d-1}r_i=\sum_{i=0}^{d-1}p_i=1$, and again inverting the inequality, we obtain
\begin{align*}
\sum_{i=0}^k r_i \geq \sum_{i=0}^k p_i , \qquad \forall \, 0\leq k \leq d-1  .
\end{align*}
The above inequality implies that $\vec{r}\succ_h \vec{p}$. Now
\begin{align*}
E(\sigma(\vec{p}))- E(\rho(\vec{r}))&=\sum_{i=0}^{d-1}E_i(p_i-r_i)\\
&=\sum_{k=0}^{d-1}(E_k- E_{k+1})\sum_{i=0}^{k}(p_i-r_i)\\
&\geq 0,
\end{align*}
with the convention $E_d=0$.
The last inequality follows from the majorization condition and the fact that $E_k\leq E_{k+1}$ for all $0\leq k\leq d-2$. The term with $k=d-1$ vanishes as the vectors $\vec{r}$ and $\vec{p}$ are normalized. This concludes the proof of the proposition.
\end{proof}

It is intuitive to see that a virtually cooler state (compared to a reference state) necessarily has a lower energy (compared to this reference state). 
Also, since $\rho(\vec{r})\succ_{vc}\sigma(\vec{p})$ implies $\vec{r}\succ_h\vec{p}$, not only the energy function is a monotone but also all Schur-concave functions are monotones. Further, the notion of virtually cooler states can be given a thermodynamical interpretation based on the setup used to demonstrate the working of a quantum refrigerator \cite{Brunner2012, Skrzypczyk2015, Silva2016}. As we prove in Appendix \ref{append:thermo}, a virtually cooler passive state can indeed be used to cool an external system to a further extent than the passive state it is compared to.

Note that the converse of Proposition \ref{prop:vc-implies-Hoffman} does not hold in general, except for qubits. For example, in a qutrit case, $\vec{r}=(0.8, 0.18, 0.02)^T$ and $\vec{p}=(0.75, 0.15, 0.1)^T$ are two decreasing vectors such that $\vec{r}\succ_h \vec{p}$
 while $\rho(\vec{r})$ is not virtually cooler than $\sigma(\vec{p})$.

Finally, let us mention the following two states which are of special importance with respect to the partial order $\succ_{vc }$ relation: (1) the ground state, which is virtually cooler than any other passive state, and (2) the maximally mixed state, with respect to which every other passive state is virtually cooler. In particular, we have the following proposition.

\begin{proposition}
\label{prop:max-mix}
The set of virtually cooler states with respect to some fixed passive state $\sigma(\vec{p})=\sum_{i=0}^{d-1}p_i\op{i}$ is a convex set. Moreover, the set of virtually cooler states with respect to the maximally mixed state is equal to the set of passive states.
\end{proposition}

\begin{proof}
Let $\rho(\vec{r})=\sum_{i=0}^{d-1}r_i\op{i}$ and $\rho(\vec{s})=\sum_{i=0}^{d-1}s_i\op{i}$ be two virtually cooler states than $\sigma(\vec{p})$. To show that the set of all virtually cooler states than $\sigma(\vec{p})$ is a convex set, it is sufficient to show that $t\, \rho(\vec{r})+\bar{t}\, \rho(\vec{s})\succ_{vc }\sigma(\vec{p})$, where $0\leq t\leq 1 $ and $\bar{t}=1-t$. We have
\begin{align*}
t\, \rho(\vec{r})+\bar{t}\, \rho(\vec{s})=\sum_{i=0}^{d-1}(t \, r_i+\bar{t} \, s_i)\op{i}.
\end{align*}
and
\begin{align*}
\frac{(t \, r_i+\bar{t} \, s_i)}{p_i}\geq \frac{(t \, r_{i+1}+\bar{t} \, s_{i+1})}{p_{i+1}}
\end{align*}
since $\rho(\vec{r})\succ_{vc}\sigma(\vec{p})$ and $\rho(\vec{s})\succ_{vc}\sigma(\vec{p})$. This concludes the proof of the first part.

For the second part, notice that all the extremal states of the set of passive states of a qu$d$it system (including the maximally mixed state) are virtually cooler than the maximally mixed state. So, from the convexity of the set of the virtually cooler states, it follows that all passive states are virtually cooler than the maximally mixed state. This concludes the proof of the second part of the proposition.  
\end{proof}

We will see in Sec.~\ref{sec:stat-tran} that the notion of relative passivity and virtually cooler states plays a critical role in the transformation of active states under incoherent operations. Before turning to this result, we need to define the notion of quantum operations that preserve relative passivity (in analogy with the quantum operations that preserve passivity).

\subsection{Relative passivity-preserving operations}
A \textit{relative passivity-preserving operation} (RPPO) is a quantum channel that is defined with respect to two fixed passive states. Given two fixed passive states $\sigma(\vec{p})$ and $\sigma(\vec{q})$, a quantum channel defines a  RPPO if it maps all passive states that are virtually cooler than $\sigma(\vec{p})$ into passive states that are virtually cooler than $\sigma(\vec{q})$. In other words, for all states $\rho(\vec{r})$ such that $\rho(\vec{r})\succ_{vc}\sigma(\vec{p})$, a quantum channel $\Lambda$ is a RPPO if $\Lambda(\rho(\vec{r}))\succ_{vc} \sigma(\vec{q})$. Let us denote the set of all RPPOs with respect to passive states $\sigma(\vec{p})$ and $\sigma(\vec{q})$ by $\mathfrak{L}_{p,q}$ (see Fig. \ref{fig:rppos}). By definition, the set $\mathfrak{L}_{p,q}$ of RPPOs preserve the partial order (on the set of passive states) of being virtually cooler than the reference passive state $\sigma(\vec{p})$ at input and $\sigma(\vec{q})$ at output.

\begin{figure}[htbp!]
\centering
\includegraphics[width=0.75\columnwidth]{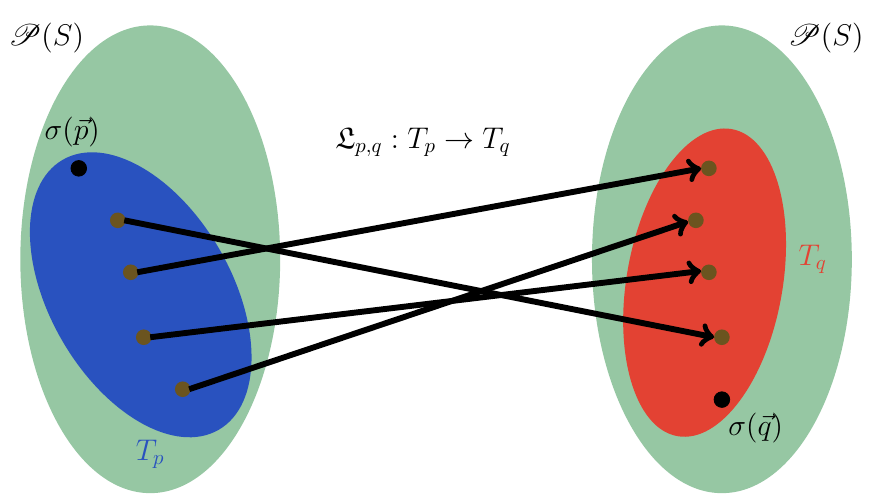}
\caption{A schematic for RPPOs. In the schematic, $\mathscr{P}(S)$ is the set of all passive states. $\sigma(\vec{p})$ and $\sigma(\vec{q})$ are two fixed passive states. $T_p$ is the set of all passive states that are virtually cooler than $\sigma(\vec{p})$ and  $T_q$ is the set of all passive states that are virtually cooler than $\sigma(\vec{q})$. That is, $T_p:=\{\rho\in\mathscr{P}(S): \rho\succ_{vc} \sigma(\vec{p})\}$ and $T_q:=\{\rho\in\mathscr{P}(S): \rho\succ_{vc} \sigma(\vec{q})\}$. Then $\mathfrak{L}_{p,q}$ denotes RPPOs and is the set of all quantum channels from the set $T_p$ into the set $T_q$.}
\label{fig:rppos}
\end{figure}

Choosing the passive states $\sigma(\vec{p})$ and $\sigma(\vec{q})$ as being the maximally mixed state $\mathbb{I}/d$ in the set of operations $\mathfrak{L}_{p,q}$ is of particular relevance as we have following proposition.
\begin{proposition}
For the choice $\sigma(\vec{p})= \sigma(\vec{q}) = \mathbb{I}/d$, the set $\mathfrak{L}_{p,q}$ is equal to the set of passivity-preserving operations.
\end{proposition}
\begin{proof}
The proof of the above proposition follows from Proposition \ref{prop:max-mix} (the passivity condition is equivalent to being virtually cooler than the maximally mixed state) and from the definition of RPPO.
\end{proof}

Thus, with this choice of $\sigma(\vec{p})$ and $\sigma(\vec{q})$, the RPPOs are equivalent to PPOs. For other choices, however, RPPOs are not necessarily passivity-preserving operations, except in their action on the ground state of the system. Indeed, it is clear that the ground state is virtually cooler than any passive state $\sigma(\vec{p})$, therefore, under the set $\mathfrak{L}_{p,q}$, the ground state is always mapped onto a state that is virtually cooler than the passive state $\sigma(\vec{q})$, hence it is passive. Thus, if the input state is the ground state, then the output state of any RPPO is always passive.

In the following, we show that the RPPOs are incoherent operations, that is, they preserve incoherent states (diagonal states in a reference basis which is the energy eigenbasis).

\begin{proposition}
For any choice of passive states $\sigma(\vec{p})$ and $\sigma(\vec{q})$, any relative passivity-preserving operation $\Lambda\in \mathfrak{L}_{p,q}$ is an incoherent operation.
\end{proposition}
\begin{proof}
We have already seen that for $\Lambda\in\mathfrak{L}_{p,q}$, $\Lambda(\op{0})$ is an incoherent state. We prove the proposition by induction. Let us assume that $\Lambda(\op{i})$ is incoherent for $i=0,\cdots,d-2$. Then following two cases arise.\\
{\bf Case 1:} The state $\sigma(\vec{p})$ is such that $p_{d-1}=0$. In this case all the states $\rho(\vec{r})$ that are virtually cooler than $\sigma(\vec{p})$ must have $r_{d-1}=0$. Therefore  we can effectively consider the transformations on the $(d-1)$ dimensional Hilbert space and by the inductive assumption $\Lambda(\op{i})$ is incoherent for $i=0,\cdots,d-2$, the operation $\Lambda$ is incoherent.\\
{\bf Case 2:} The state $\sigma(\vec{p})$ is such that $p_{d-1}\neq0$. Let us consider a state $\rho(\vec{r})$ such that $r_i=p_i$ for $1\leq i\leq (d-2)$, $r_{0}=p_{0}+\epsilon$ and $r_{d-1}=p_{d-1}-\epsilon$ with $\epsilon>0$. Thus $\rho(\vec{r})\succ_{vc}\sigma(\vec{p})$. By definition, $\Lambda(\rho(\vec{r}))$ is a passive, therefore, incoherent state. Then,
\begin{align*}
&\Lambda(\rho(\vec{r}))-\Lambda(\sigma(\vec{p}))\\
&=\epsilon\left[\Lambda(\op{0})-\Lambda(\op{d-1})\right].
\end{align*}
Since $\Lambda(\op{0})$ and LHS of above equation both are incoherent, we get that $\Lambda(\op{d-1})$ is also incoherent, i.e., diagonal in the energy eigenbasis. This concludes the proof of the proposition.
\end{proof}

Let us now present an instructive example of a RPPO for the qutrit case.  Let us choose $\sigma(\vec{p})=\mathrm{diag}(p_0,p_1,p_2)$ and $\sigma(\vec{q})=\mathrm{diag}(q_0,q_1,q_2)$, where $p_0\geq p_1\geq p_2>0$, $q_0\geq q_1\geq q_2>0$, and  $\sum_{i=0}^2p_i=\sum_{i=0}^2q_i=1$. Further, choose  $p_0=p_1=(q_0+q_1)/2$ and $p_2=q_2$, so that $\vec{p}\prec_{h}\vec{q}$. Then, the quantum channel $\Lambda$ that is defined with Kraus operators $K_1$ and $K_2$,
\begin{align*}
&K_1=\begin{pmatrix}
\sqrt{\frac{q_0}{2p_0}} & 0 &0\\
0 & \sqrt{\frac{q_1}{2p_1}} & 0\\
0 & 0 & \sqrt{\frac{q_2}{2p_2}}
\end{pmatrix}, \\
&K_2=\begin{pmatrix}
0 & \sqrt{\frac{q_0}{2p_1}} &0\\
\sqrt{\frac{q_1}{2p_0}} & 0 & 0\\
0 & 0 & \sqrt{\frac{q_2}{2p_2}}
\end{pmatrix},
\end{align*}
is RPPO with respect to passive states $\sigma(\vec{p})$ and $\sigma(\vec{q})$, as we show in the following. That the map $\Lambda$ is trace preserving follows from the relation between $\vec{p}$ and $\vec{q}$, namely
\begin{align*}
\sum_{i=1}^{2}K_i^\dagger K_i &= \frac{q_0 + q_1}{2 p_0} \op{0}  +\frac{q_0 + q_1}{2 p_1} \op{1} + \frac{q_2}{p_2} \op{2}  \\
&= \openone.
\end{align*}
The action of such a map on a passive state $\rho(\vec{r})=\mathrm{diag}(r_0,r_1,r_2)$ with $r_0\geq r_1\geq r_2>0$ and $\sum_{i=0}^2 r_i=1$, is given by
\begin{align*}
&\Lambda(\rho(\vec{r}))= \sum_{i=1}^{2}K_i\rho(\vec{r})K_i^\dagger\\
&=\left(\frac{r_0}{2p_0}+\frac{r_1}{2p_1}\right)(q_0\op{0} + q_1\op{1})+ \frac{r_2}{p_2}q_2\op{2}\\
&:=s_0\op{0}+s_1\op{1}+s_2\op{2}.
\end{align*}
Using above equation, we can write
\begin{align*}
\begin{pmatrix}
s_0/q_0\\
s_1/q_1\\
s_2/q_2
\end{pmatrix}=\underbrace{ \frac{1}{2} \begin{pmatrix}
1 & 1 & 0\\
1 & 1 & 0\\
0 & 0 & 2
\end{pmatrix} }_{R}
\begin{pmatrix}
r_0/p_0\\
r_1/p_1\\
r_2/p_2
\end{pmatrix}.
\end{align*}
Since we recognize that $R$ is a Hoffman matrix, the above equation implies that if $\rho(\vec{r})\succ_{vc}\sigma(\vec{p})$ then $\Lambda(\rho(\vec{r}))\succ_{vc}\sigma(\vec{q})$, so the quantum channel $\Lambda$ is a RPPO.

\subsection{Hierarchy of various quantum operations}
It is instructive to compare the various sets of quantum operations considered so far, namely, the passivity-preserving operations (PPO), the relative passivity-preserving operations (RPPO), and incoherent operations (those that preserve incoherent states in the energy eigenbasis). First, note that for an arbitrary choice of the passive states $\sigma(\vec{p})$ and $\sigma(\vec{q})$, RPPOs only preserve the passivity of virtually cooler states than $\sigma(\vec{p})$ since these are mapped onto virtually cooler states than $\sigma(\vec{q})$, which are themselves passive. This means that RPPOs are not passivity preserving, in general. Similarly, all PPOs are not RPPOs for some choice of passive states $\sigma(\vec{p})$ and $\sigma(\vec{q})$. Thus $\mathfrak{L}_{p,q}\not\subset \mathrm{PPO}$ and $\mathrm{PPO}\not\subset \mathfrak{L}_{p,q}$, where PPO denotes here the set of passivity-preserving operations. Similar relationships hold if we consider strictly incoherent RPPOs and strictly incoherent PPOs. In contrast, we know that both PPO and RPPO sets are necessarily inside the set of incoherent operations. We refer to Fig. \ref{fig:ops} for more details of these relationships as a schematic diagram.

\begin{figure}[htbp!]
\centering
\includegraphics[width=0.61\columnwidth]{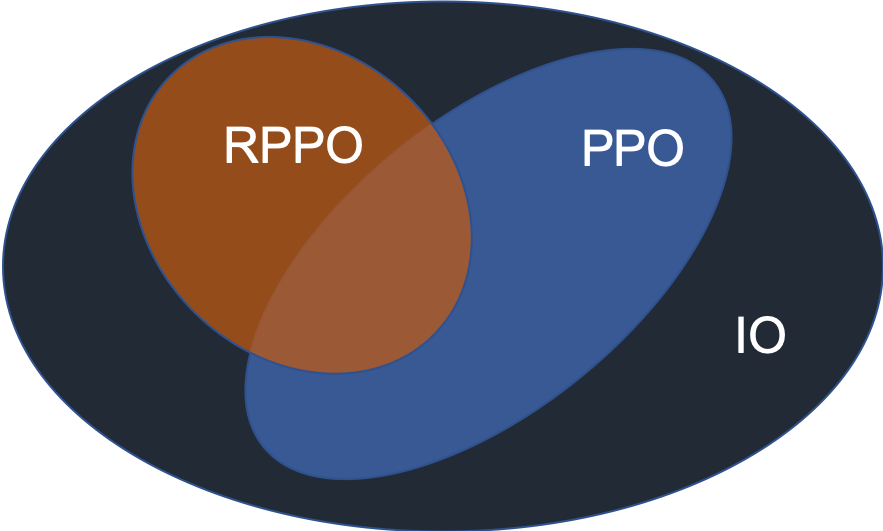}
\caption{In above schematic diagram, RPPO is the set of relative passivity preserving operations for two fixed passive state $\sigma(\vec{p})$ and $\sigma(\vec{q})$, PPO is the set of passivity preserving operations, and  IO is the set of incoherent operations. If we choose $\sigma(\vec{p})$ and $\sigma(\vec{q})$ both to be  equal to the maximally mixed state, then RPPO is equal to PPO.}
\label{fig:ops}
\end{figure}

\section{Transformation of pure active states under relative passivity-preserving \\ quantum operations}
\label{sec:stat-tran}
Building on the notions of relative passive states and relative passivity-preserving operations (RPPOs), we are now ready to consider the problem of  pure active state transformations under the set $\mathfrak{L}_{p,q}$ of RPPOs.  Our central result is to find the necessary and sufficient condition for such a transformation to exist, based on Hoffman majorization.

\

\subsection{System of arbitrary dimension}
Let us consider the set $\mathfrak{D}$ of pure states that have the form $\ket{\psi}=\sum_{i=1}^{d}\sqrt{p_i}\, e^{i\theta_i}\ket{i}$ with $\{\theta_i\}\in\mathbb{R}$ and $\{p_i\}\in \mathcal{S}(d)$, that is $p_i\leq p_j$ for all $i> j$, and $p_0\in (0,1)$. The set $\mathfrak{D}$ is a subset of the class of active states, and we note that each state in $\mathfrak{D}$ can be associated with a passive state $\sigma(\vec{p})$ having the same population of energy eigenstates. In the following, we provide a necessary and sufficient condition for the transformation between a pair of pure states in $\mathfrak{D}$ under RPPO relative to their respective passive distributions.


\begin{theorem}
\label{th:qud-iff}
An active state $\ket{\psi} \in \mathfrak{D}$ can be transformed to another active state $\ket{\phi}\in \mathfrak{D}$ under a strictly incoherent RPPO with respect to $\sigma(\vec{p})$ and $\sigma(\vec{q})$ if and only if $\vec{p}\prec_h \vec{q}$, where $\ket{\psi}=\sum_{i=0}^{d-1}\sqrt{p_i}e^{i\theta_i}\ket{i}$, $\ket{\phi}=\sum_{i=0}^{d-1}\sqrt{q_i}e^{i\nu_i}\ket{i}$, and $\vec{p}$ and  $\vec{q}$ are passive vectors.
\end{theorem}

\begin{proof}
Before we begin our proof, let us remark that for the purpose of state transformations under RPPOs, the states $\ket{\psi}=\sum_{i=0}^{d-1}\sqrt{p_i}\ket{i}$ and $\ket{\psi'}=\sum_{i=0}^{d-1}\sqrt{p_i}e^{i\theta_i}\ket{i}$ are equivalent as $\ket{\psi'} = \mathrm{diag}(e^{i\theta_0},\cdots, e^{i\theta_{d-1}}) \ket{\psi}$, where $U(\vec{\theta})\coloneqq \mathrm{diag}(e^{i\theta_0},\cdots, e^{i\theta_{d-1}})$ is a unitary operator corresponding to RPPO. Unitary operation $\mathcal{U}_{\vec{\theta}}(\cdot)\coloneqq U(\vec{\theta})(\cdot)U^\dag (\vec{\theta}) $ is such that $\Lambda\circ\mathcal{U}_{\vec{\theta}}, \mathcal{U}_{\vec{\theta}}\circ\Lambda\in\mathfrak{L}_{p,q}$ for all $\Lambda\in \mathfrak{L}_{p,q}$. So, without loss of generality, we can drop all the phases from $\ket{\psi}$ and $\ket{\phi}$ in the above theorem. 

More formally, we can say that active states $\ket{\psi(\vec{\theta})}\in\mathfrak{D}$ in the set $\left\{\ket{\psi(\vec{\theta})}\right\}_{\vec{\theta}}$ form an equivalence class if they can be transformed to one another by an energy-preserving unitary operation $\mathcal{U}\in\mathfrak{L}_{p,p}$ (note here that the input and output reference probabilities are equal to $\vec{p}$). For example, $\ket{\psi}=\sum_{i=0}^{d-1}\sqrt{p_i}\ket{i}$ and $\ket{\psi'}=\sum_{i=0}^{d-1}\sqrt{p_i}e^{i\theta_i}\ket{i}$ belong to the same equivalence class as $\ket{\psi'} = \mathrm{diag}(e^{i\theta_0},\cdots, e^{i\theta_{d-1}})\ket{\psi}$, where $\mathrm{diag}(e^{i\theta_0},\cdots, e^{i\theta_{d-1}})$ is unitary operator corresponding to unitary RPPO in $\mathfrak{L}_{p,p}$.

{\it Proof of ``if" part.--} We will prove the theorem using mathematical induction. Let us start with  $\dim \mc{H}=d=2$. If $p_1=0$, the Hoffman majorization implies that $q_1=0$. Therefore $\ket{\psi}=\ket{\phi} =\ket{0}$ and the transformation is trivial via the identity operation. Let us assume that $p_1\neq 0$. Now, since $\begin{pmatrix}
p_0\\
p_1\end{pmatrix} \prec_h \begin{pmatrix}
q_0\\
q_1\end{pmatrix}$, there exists a Hoffman matrix $R$ such that 
\begin{align}
\begin{pmatrix}
p_0\\
p_1\end{pmatrix}=R\begin{pmatrix}
q_0\\
q_1\end{pmatrix},
\end{align}
where $R=\begin{pmatrix}
a & \bar{a}\\
\bar{a} & a
\end{pmatrix}
$ with $1/2\leq a\leq 1$ and $\bar{a}=1-a$. It is easy to check that the operation $\Phi_p$ defined via Kraus operators $L_1$ and $L_2$ (represented in energy eigenbasis)
\begin{align}
L_1=\begin{pmatrix}
\sqrt{a\frac{q_0}{p_0}}&
0\\0&\sqrt{a\frac{q_1}{p_1}}
\end{pmatrix};~L_2=\begin{pmatrix}
 0 &\sqrt {\bar{a}\frac{q_0}{p_1}}\\
 \sqrt {\bar{a}\frac{q_1}{p_0}}&0
 \end{pmatrix}
\end{align}
maps $\ket{\psi}$ to $\ket{\phi}$. We need to show that this map $\Phi_p$ is a RPPO. We have
\begin{align}
&\Phi_p(\op{0})
=\frac{1}{p_0}\left(R_{0,0} q_0 \op{0}+  R_{0,1} q_1\op{1}\right) \\
&\Phi_p(\op{1})= \frac{1}{p_1}\left(R_{1,0} q_0 \op{0}+  R_{1,1} q_1\op{1}\right).
\end{align}
Now, let us consider an input state $\rho(\vec{r})=r_0\op{0}+r_1\op{1}$ that is virtually cooler than the passive state $\sigma(\vec{p})=p_0\op{0}+p_1\op{1}$ associated with $\ket{\psi}$. Then, the corresponding output state reads
\begin{align*}
\rho(\vec{s}):=\Phi_p(\rho(\vec{r})) = \sum_{i=0}^1 \sum_{j=0}^1 \frac{r_i R_{i,j} q_j}{p_i}\op{j}.
\end{align*}
implying that 
\begin{align*}
\frac{s_j}{q_j} = \sum_{i=0}^1\frac{r_i R_{i,j} }{p_i}
\end{align*}

Since $R$ is a (symmetric) Hoffman matrix, the condition  $\rho(\vec{r})\succ_{vc}\sigma(\vec{p})$ (i.e., $r_0 / p_0\geq r_1 / p_1$) implies the condition that  $\rho(\vec{s})\succ_{vc}\sigma(\vec{q})$ (i.e. $s_0 / q_0\geq s_1 / q_1$) where $\sigma(\vec{q})=q_0\op{0}+q_1\op{1}$ is the passive state associated with  $\ket{\phi}$. Therefore $\Phi_p$ maps the state $\rho(\vec{r})$, which is virtually cooler than $\sigma(\vec{p})$, to a state $\rho(\vec{s})$, which is virtually cooler than $\sigma(\vec{q})$. Also, from the form of Kraus operators, it is clear that $\Phi_p$ is a strictly incoherent operation.  Now, assuming that the theorem holds true for $\dim \mc{H}\leq d-1$, we will prove that the theorem holds true for $\dim \mc{H}=d$. Now, the two cases arise.

\bigskip
\noindent
{\it Case 1.} Let there be a $k$ ($0< k< d-1$) such that $p_k\neq 0$ and $p_{k+1}=\cdots=p_{d-1}=0$. From Hoffman majorization, it follows that $q_{k+1}=\cdots=q_{d-1}=0$. Now, let the vector
$(p_0, \cdots, p_k)^T \prec_h (q_0,\cdots, q_k)^T$. Then, from the inductive assumption, we know there exits a RPPO $\Phi_p'$ on $\msc{D}_k$ (the set of all $k\times k$ density matrices) specified by the Kraus operators $\{L_\mu\}$, $\mu=0,\cdots, N-1$ such that $\sum_{i=0}^{k-1} \sqrt{p_i}\ket{i}\xrightarrow{\Phi_p'}\sum_{i=0}^{k-1} \sqrt{q_i}\ket{i}$. Let us consider $K_\mu=L_\mu\oplus \frac 1 {\sqrt N}\mathbbm{I}_{d-k}$, then $\Phi_p(\cdot)=\sum_{\mu=0}^{N-1} K_\mu\cdot K_\mu^\dag$ is a RPPO that transforms $\ket{\psi}$ to $\ket{\phi}$.

\bigskip
\noindent
{\it Case 2.} When $p_{d-1}\neq 0$. Since $\vec{p}\prec_h \vec{q}$, from Theorem \ref{th:hoff2}, there exists a Hoffman matrix $R\in \mc{R}(d)$ such that $\vec{p}=R\vec{q}$. From Theorem \ref{th:hoff1}, we know that $R=\sum_{\tau\in\mathcal{P}(d)}\alpha_\tau M^{\tau}$, where $\tau$ is a partition, $\alpha_\tau$ are probabilities and $\sum_{\tau\in\mathcal{P}(d)} \alpha_\tau=1$. Each $M^{\tau}$ is a $d\times d$ matrix that corresponds to a partition $\tau=(\tau_1,\tau_2,\cdots \tau_k)$ and can be written as $M^{\tau}=\oplus_{i=1}^{k}M_{\tau_i}$, where $M_{\tau_i}=\frac{\mathfrak{I}_{|\tau_i|}}{|\tau_i|}$ and $|\tau_i|$ is the cardinality of the part $\tau_i$. Let us partition the vectors $\vec{p}$ and $\vec{q}$ following the same partition $\tau=(\tau_1,\tau_2,\cdots \tau_k)$, that is
\begin{align*}
\vec{p} = \left(\vec{p}_{\tau_1},\cdots,\vec{p}_{\tau_k}\right); \quad \vec{q} = \left(\vec{q}_{\tau_1},\cdots,\vec{q}_{\tau_k}\right).
\end{align*}
In this notation,
\begin{align*}
M^{\tau}\vec{q} &= \left(M_{\tau_1}\oplus \cdots \oplus M_{\tau_k}\right)\left(\vec{q}_{\tau_1},\cdots,\vec{q}_{\tau_k}\right)\\
&= \left(M_{\tau_1}\vec{q}_{\tau_1},\cdots,M_{\tau_k}\vec{q}_{\tau_k}\right)\\
&:= \left(\vec{r}_{\tau_1},\cdots,\vec{r}_{\tau_k}\right) = \vec{r}^{~\tau}.
\end{align*}
Further, let us similarly divide the total Hilbert space as 
\begin{align}
\mc{H}=\oplus_{i=1}^k\mc{H}_{\tau_i},
\end{align}
Now, let us build the following set of operators $\{G_{a_i}^{\tau_i}\}$, where, for each $i$, the $|\tau_i|$ operators $G_{a_i}^{\tau_i}$ act on subspace $\mc{H}_{\tau_i}$ and are labeled by $a_i\in\{0,\cdots,|\tau_i|-1\}$. These operators $\{G_{a_i}^{\tau_i}\}$ will be used to define Kraus operators.
Letting  $\mu_1=0$ and $\mu_i=\sum_{l=1}^{i-1}|\tau_l|$ for $i>1$, we define
\begin{align}
\label{eq:def-G}
G_{a_i}^{\tau_i}&=\sum_{j=0}^{|\tau_i|-1} \sqrt{\frac{q_{\mu_i+\pi_{a_i}^{\tau_i}(j)}}{p_{\mu_i+j}}}\ket{{\mu_i+{\pi_{a_i}^{\tau_i}(j)}}}\bra{{\mu_i+j}},
\end{align}
where $\{\pi^{\tau_i}_{a_i}\}$ are cyclic permutations of indices from the set $\{0,\cdots, |\tau_i|-1\}$ and each permutation is labeled by  $a_i$.
For $j=\{0,\cdots, |\tau_i|-1\}$, we thus have 
\begin{align}
\label{eq:def-G2}
G_{a_i}^{\tau_i}\ket{{\mu_i+j}}= \sqrt{\frac{q_{\mu_i+\pi^{\tau_i}_{a_i}(j)}}{p_{\mu_i+j}}}\ket{{\mu_i+{\pi_{a_i}^{\tau_i}(j)}}}.
\end{align}
Let us define
\begin{align}
\ket{\psi}=\sum_{i=1}^{k}\ket{\psi_i};~\ket{\phi}=\sum_{i=1}^{k}\ket{\phi_i},
\end{align}
where each $\ket{\psi_i}$ has disjoint support over $|\tau_i|$-dimensional disjoint subspaces $\mc{H}_{\tau_i}$. The same is true for each $\ket{\phi_i}$. Now, we see that
\begin{align*}
G_{a_i}^{\tau_i}\ket{\psi}=\ket{\phi_i}.
\end{align*}
Next, we define a map $\Phi_p$ with Kraus operators
\begin{align}
G_{\vec{a}}^\tau=\sqrt{\frac{\alpha_\tau}{|\tau|}} \sum_{i=1}^k G_{a_i}^{\tau_i},
\end{align}
where $\vec{a}=\left(a_1,\cdots,a_k\right)$, $a_i\in\{0,\cdots, |\tau_i|-1\}$, $|\tau|=\prod_{i=1}^k |\tau_i|$ and each operator in above sum has disjoint support on total Hilbert space. 
From the form of the Kraus operators, it is clear that they map incoherent states to incoherent states, therefore $\Phi_p$ constitutes a strictly incoherent operation. Now, it is easy to see that
\begin{align*}
G_{\vec{a}}^{\tau}\ket{\psi} 
&=\sqrt{\frac{\alpha_\tau}{|\tau|}}\sum_{i=1}^k G_{a_i}^{\tau_i}\ket{\psi}=\sqrt{\frac{\alpha_\tau}{|\tau|}}\sum_{i=1}^k \ket{\phi_i}\nonumber\\
&=\sqrt{\frac{\alpha_\tau}{|\tau|}} \ket{\phi}  ,
\end{align*}
implying that 
\begin{align}
\label{eq:def-G-cond2}
 \Phi_p(\op{\psi}) = \sum_\tau\sum_{\vec{a}}G_{\vec{a}}^{\tau}\op{\psi}G_{\vec{a}}^{{\tau}\dagger}=\op{\phi}.
\end{align}
The complete positivity of above map is guaranteed as it is presented in terms of Kraus operators. Now, we show that it is also trace preserving. For this, we need to show $\sum_{\tau}\sum_{\vec{a}} G_{\vec{a}}^{\tau\dagger} G_{\vec{a}}^\tau = \mathbb{I}_d$. We have indeed
\begin{align*}
&\sum_{\tau}\sum_{\vec{a}} G_{\vec{a}}^{\tau\dagger} G_{\vec{a}}^\tau\\
&=\sum_{\tau}\sum_{\vec{a}} \frac{\alpha_\tau}{|\tau|} \left(\sum_{i=1}^k G_{a_i}^{\tau_i\dagger}\right) \left(\sum_{l=1}^k G_{a_l}^{\tau_l}\right)\\
&=\sum_{\tau}\frac{\alpha_\tau}{|\tau|}\sum_{\vec{a}}  \sum_{i=1}^k G_{a_i}^{\tau_i\dagger} G_{a_i}^{\tau_i}\\
&=\sum_{\tau}\frac{\alpha_\tau}{|\tau|}\sum_{i=1}^k  \frac{|\tau|}{|\tau_i|}  \sum_{a_i=0}^{|\tau_i|-1} G_{a_i}^{\tau_i\dagger} G_{a_i}^{\tau_i} \\
&=\sum_{\tau} \alpha_\tau \sum_{i=1}^k \frac{1}{|\tau_i|} \sum_{j=0}^{|\tau_i|-1} \frac{ \sum_{a_i}q_{\mu_i+\pi_{a_i}^{\tau_i}(j)}}{p_{\mu_i+j}}\ket{{\mu_i+j}}\!\bra{{\mu_i+j}}\\
\end{align*}
\begin{align*}
&=\sum_{\tau} \alpha_\tau \sum_{i=1}^k \sum_{j=0}^{|\tau_i|-1} \frac{(\vec{r}_{\tau_i})_j}{(\vec{p}_{\tau_i})_j}\ket{{\mu_i+j}}\!\bra{{\mu_i+j}}~~~~~~~~~~~~~~~~~~~~\\
&=\sum_{\tau} \alpha_\tau \sum_{l=0}^{d-1} \frac{(\vec{r}^{~\tau})_l}{p_{l}}\ket{{l}}\!\bra{{l}}\\
&=\sum_{l=0}^{d-1} \ket{l}\bra{l}= \mathbb{I}_d,
\end{align*}
where we have used $\sum_{\tau} \alpha_\tau \vec{r}^{~\tau} = \vec{p}$ 
as well as the fact that summing over all cyclic permutations $\{\pi^{\tau_i}_{a_i}\}$ from the set $\{0,\cdots, |\tau_i|-1\}$ results in
\begin{align}
\sum_{a_i=0}^{|\tau_i|-1}q_{\mu_i+\pi^{\tau_i}_{a_i}(j)} = |\tau_i| \, (\vec{r}_{\tau_i})_j.
\end{align}
Therefore the map $\Phi_p$ with Kraus operators $\{G_{\vec{a}}^\tau\}$ forms a completely positive trace-preserving map. The final step is to prove that $\Phi_p$ is a RPPO with respect to the reference passive state $\sigma(\vec{p})$ associated with $\ket{\psi}$ at the input and reference passive state $\sigma(\vec{q})$ associated with $\ket{\phi}$ at the output. For a given $\tau$ and for all $\ket{{j}} \in \mc{H}_{\tau_i}$, i.e., $\ket{{\mu_i+j}}$, from Eq. \eqref{eq:def-G2}, we have
\begin{align*}
&\sum_{\vec{a}}G_{\vec{a}}^{\tau}\ket{{\mu_i+j}} \bra{{\mu_i+j}}G_{\vec{a}}^{\tau\dagger}\nonumber\\
&=\frac{\alpha_\tau}{|\tau|}\sum_{\vec{a}}  \left(\sum_{l=1}^{k}G_{a_l}^{\tau_l}\right) \ket{{\mu_i+j}}\! \bra{{\mu_i+j}} \left(\sum_{m=1}^{k}G_{a_m}^{\tau_m\dagger}\right)  \nonumber\\
&=\frac{\alpha_\tau}{|\tau|}\sum_{\vec{a}}  G_{a_i}^{\tau_i} \ket{{\mu_i+j}}\! \bra{{\mu_i+j}} G_{a_i}^{\tau_i\dagger}  \nonumber\\
&=\frac{\alpha_\tau}{|\tau|}\sum_{\vec{a}} \frac{q_{\mu_i+\pi_{a_i}^{\tau_i}(j)}}{p_{\mu_i+j}} \ket{{\mu_i+{\pi_{a_i}^{\tau_i}(j)}}} \! \bra{{\mu_i+{\pi_{a_i}^{\tau_i}(j)}}}\nonumber\\
&=\frac{\alpha_\tau}{|\tau_i|}\sum_{a_i=0}^{|\tau_i|-1} \frac{q_{\mu_i+\pi_{a_i}^{\tau_i}(j)}}{p_{\mu_i+j}} \ket{{\mu_i+{\pi_{a_i}^{\tau_i}(j)}}}\!  \bra{{\mu_i+{\pi_{a_i}^{\tau_i}(j)}}}\nonumber\\
&=\frac{\alpha_\tau}{|\tau_i|}\sum_{l=0}^{|\tau_i|-1} \frac{q_{\mu_i+l}}{p_{\mu_i+j}} \ket{{\mu_i+l}} \! \bra{{\mu_i+l}}\nonumber\\
&=\frac{\alpha_\tau}{(\vec{p}_{\tau_i})_j}\sum_{l=0}^{|\tau_i|-1} (M_{\tau_i})_{j,l} \,  (\vec{q}_{\tau_i})_l \ket{{\mu_i+l}} \! \bra{{\mu_i+l}}  .
\end{align*}
Therefore, for an input state $\ket{i}$, we have
\begin{align*}
\Phi_P(\op{i})&=\sum_\tau \sum_{\vec{a}}G_{\vec{a}}^{\tau}\op{i} G_{\vec{a}}^{\tau\dagger}\\
&=\sum_\tau   \frac{\alpha_\tau}{p_i}  \sum_{j=0}^{d-1}    (M^\tau)_{i,j} \, q_j   \op{j} \\
&=\frac{1}{ p_i}\sum_{j=0}^{d-1} R_{i,j} \, q_{j} \op{j},
\end{align*}
where we have used $\sum_\tau  \alpha_\tau M^\tau = R$. 
Now, consider an arbitrary input passive state $\rho(\vec{r}) =\sum_{i=0}^{d-1} r_i\op{i}$ that is virtually cooler than the passive state $\sigma(\vec{p}) =\sum_{i=0}^{d-1} p_i\op{i}$. Then, at the output, we have
\begin{align*}
&\rho(\vec{s}):=\Phi_P\big(\rho(\vec{r})\big)=\sum_{j=0}^{d-1}\left(\sum_{i=0}^{d-1} \frac{r_i R_{i,j} }{ p_i} \right)q_{j} \op{j}.
\end{align*}
From the above equation, we have
\begin{align*}
\frac{s_j}{q_j} =\sum_{i=0}^{d-1} \frac{r_i R_{i,j} }{ p_i}   .
\end{align*}
Since $R$ is a (symmetric) Hoffman matrix, it follows that 
\begin{align*}
\frac{r_i}{p_i} \geq \frac{r_{i+1}}{p_{i+1}} \Rightarrow \frac{s_i}{q_i} \geq \frac{s_{i+1}}{q_{i+1}}.
\end{align*}
Thus the map $\Phi_p$ maps a passive state $\rho(\vec{r})$ that is virtually cooler than $\sigma(\vec{p})$ onto a passive state $\rho(\vec{s})$ that is virtually cooler than $\sigma(\vec{q})$. Hence, $\Phi_{p}$ is a RPPO. This concludes the proof of the ``if'' part of the theorem. 

\medskip
\noindent
{\it Proof of ``only if" part.--} Let us assume that there exists a strictly incoherent RPPO $\Phi_p$ such that $\op{\phi}=\Phi_p(\op{\psi})$. Since it is a strictly incoherent operation, from Refs.~\cite{Du2015, DuE2017}, we conclude that $\vec{p}\prec \vec{q}$. Given that $\vec{p}$ and $\vec{q}$ are passive vectors, this in turn implies that $\vec{p}\prec_h \vec{q}$ and concludes the proof of the theorem.
\end{proof}

It must be noted that the RPPO $\Phi_p$ that we have constructed in the proof satisfies $\Phi_p(\sigma(\vec{p}))=\sigma(\vec{q})$, in addition to the requested property
$\Phi_p(\op{\psi}) = \op{\phi}$. This makes sense since $\Phi_p$ is (strictly) incoherent and since $\sigma(\vec{p})$ is the incoherent version of $\ket{\psi}$ and $\sigma(\vec{q})$ is the incoherent version of $\ket{\phi}$. Now, we describe some consequences of the above theorem.

\begin{corollary}
\label{rem:conv-pure1}
The state $\ket{\psi}_{\max}=\frac{1}{\sqrt{d}}\sum_{i=0}^{d-1}e^{i\theta_i}\ket{i}$ can be transformed to any state $\ket{\phi}=\sum_{i=0}^{d-1}\sqrt{q_i}\, e^{i\nu_i}\ket{i}\in \mathfrak{D}$ using strictly incoherent RPPOs defined with respect to vectors $\vec{p}=\{d^{-1},\cdots,d^{-1}\}^T$ and $\vec{q}$.
\end{corollary}
\begin{proof}
Since, the vector $\vec{p}=\{d^{-1},\cdots,d^{-1}\}^T\prec_h \vec{q}$ for any choice of $\vec{q}$, the corollary follows from Theorem \ref{th:qud-iff}. Thus, in the set $\mathfrak{D}$, $\ket{\psi}_{\max}$ are the most active states. They can be converted into any active state $\ket{\phi}\in \mathfrak{D}$. Since $\sigma(\vec{p})$ is the maximally mixed state, the condition $\rho(\vec{r}) \succ_{vc} \sigma(\vec{p})$ boils down to the condition that $\rho(\vec{r})$ is passive. Hence, the RPPO transforming $\ket{\psi}_{\max}$ into $\ket{\phi}$ is such that it transforms any passive state into a state that is virtually cooler that $\sigma(\vec{q})$, i.e., the passive state associated to $\ket{\phi}$.
\end{proof}
\begin{corollary}
\label{rem:conv-pure2}
Any state $\ket{\psi}=\sum_{i=0}^{d-1}\sqrt{p_i}\, e^{i\gamma_i}\ket{i}\in \mathfrak{D}$ can be transformed to the ground state $\ket{0}$ using strictly incoherent RPPOs defined with respect to vectors $\vec{p}$ and $\vec{q}=\{1,0,\cdots,0\}^T$. 
\end{corollary}
\begin{proof}
Noting that $\vec{p} \prec_h \vec{q}=\{1,0,\cdots,0\}^T$ for any vector $\vec{p}$, the corollary follows from Theorem \ref{th:qud-iff}. In the set $\mathfrak{D}$, the ground state $\ket{0}$  is the least active state; actually, it is the only state in $\mathfrak{D}$ that is passive. Since $\sigma(\vec{q})=\op{0}$, the condition $\rho(\vec{s}) \succ_{vc} \sigma(\vec{q})$ at the output implies that $\rho(\vec{s})=\op{0}$ (remember that the ground state is virtually cooler than every other state). Hence, the RPPO transforming $\ket{\psi}$ into $\ket{0}$ is such that any input passive state that is virtually cooler than $\sigma(\vec{p})$, i.e., the passive state associated to $\ket{\psi}$, must be converted into the ground state $\ket{0}$.
\end{proof}

\subsection{Special case of qubit systems}
Let us now examine the qubit case in more details. It is known that for qubit systems, the set of strictly incoherent operations and the set of incoherent operations coincide \cite{Chitambar2016}. In contrast, it is easy to see that for fixed qubit states $\sigma(\vec{p})\neq \mathbb{I}/2$ and $\sigma(\vec{q})\neq \mathbb{I}/2$, the set of RPPOs is not equal to the set of PPOs. However, it turns out that pure qubit active state transformations can be achieved with passivity-preserving operations (there is no need to consider relative passivity-preserving operations). The following theorem proves this assertion.

\begin{theorem}
\label{th:qubit-iff}
A qubit active state $\ket{\psi(\vec{\theta})} \in \mathfrak{D}$ can be transformed to another qubit active state $\ket{\phi(\vec{\nu})}\in \mathfrak{D}$ under a PPO  if and only if $\vec{p}\prec_h \vec{q}$, where $\ket{\psi(\vec{\theta})}=\sum_{i=0}^{1}\sqrt{p_i}e^{i\theta_i}\ket{i}$, $\ket{\phi(\vec{\nu})}=\sum_{i=0}^{1}\sqrt{q_i}e^{i\nu_i}\ket{i}$, $\vec{p}=\{p_0,p_1\}$ with $p_0\geq p_1$, and  $\vec{q}=\{q_0,q_1\}$ with $q_0\geq q_1$.
\end{theorem}
\begin{proof}
Let us first notice that states $\ket{\psi(\vec{\theta})}$ and $\ket{\phi(\vec{\nu})}$ are equivalent to states $\ket{\psi}=\sum_{i=0}^{1}\sqrt{p_i}\ket{i}$ and  $\ket{\phi}=\sum_{i=0}^{1}\sqrt{q_i}\ket{i}$, respectively. Therefore it suffices to prove above theorem for states $\ket{\psi}$ and $\ket{\phi}$. 

We know that the passivity-preserving operations are incoherent operations and that incoherent operations are equivalent to strictly incoherent operations for the qubit case. Then, the existence of a passivity-preserving operation transforming $\ket{\psi}$ to $\ket{\phi}$ implies that $\vec{p}\prec_h \vec{q}$ based on Refs. \cite{Du2015, DuE2017}, which concludes the ``only if'' part of the theorem.

For the ``if part'', we need to construct a map that transforms $\ket{\psi}$ to $\ket{\phi}$ using the condition $\vec{p}\prec_h \vec{q}$ in the same way as we constructed in the proof of Theorem \ref{th:qud-iff}. We apply this map on the two extremal passive states for a qubit, namely
\begin{align}
\sigma_0^p=  \begin{pmatrix}
1 & 0\\
0 & 0
\end{pmatrix} \qquad  \sigma_1^p= \begin{pmatrix}
\frac{1}{2} & 0\\
0 & \frac{1}{2}
\end{pmatrix} 
\end{align}
For the first extremal state, we have
\begin{align}
&\Phi_p(\sigma_0^p)
=\begin{pmatrix}
 \frac{a q_0}{p_0} & 0\\
0 &  \frac{\bar{a}q_1}{p_0}
\end{pmatrix}.
\end{align}
The above state is passive as $a\geq \bar{a}$, $q_0\geq q_1$, and thus $a q_0\geq \bar{a}q_1$. Similarly, for the second extremal state, we have
\begin{align}
&\Phi_p(\sigma_1^p)
= \begin{pmatrix}
\frac{a q_0}{2p_0} + \frac{\bar{a} q_0}{2p_1} & 0\\
0 & \frac{\bar{a}q_1}{2p_0} + \frac{a q_1}{2p_1}
\end{pmatrix}.
\end{align}
To show that $\Phi_p(\sigma_1^p)$ is a passive state, let us consider
\begin{align*}
\frac{a q_0}{2p_0} + \frac{\bar{a} q_0}{2p_1} -\frac{\bar{a}q_1}{2p_0} - \frac{a q_1}{2p_1}
&\geq \frac{q_0}{2p_0} - \frac{ q_1}{2p_1}\\
&=\frac{q_0 p_1 - p_0 q_1}{2p_0 p_1}\geq 0.
\end{align*}
where the first inequality comes from $p_0 \ge p_1$. In the second inequality, 
we have used the fact that $\vec{p}\prec_h \vec{q}$ implies $q_0\ge p_0$ and $q_1\le p_1$ .  Since all passive states are convex mixtures of $\sigma_0^p$ and $\sigma_1^p$, this shows that the operation $\Phi_p$ is a passivity-preserving operation.
\end{proof}

\subsection{Maximal extractable work from RPPOs}
First, let us consider the case of pure active states from the set $\mathfrak{D}$. For such states, we have following proposition.
\begin{proposition}
The maximal amount of work that can be extracted from any state $\ket{\psi}=\sum_{i=0}^{d-1}\sqrt{p_i}\, e^{i\gamma_i}\ket{i}$ in the set $\mathfrak{D}$ using strictly incoherent RPPOs defined with respect to vectors $\vec{p}$ and $\vec{q}=\{1,0,\cdots,0\}^T$ is equal to the ergotropy of the state.
\end{proposition}
\begin{proof}
Let us consider an arbitrary state $\ket{\psi}=\sum_{i=0}^{d-1}\sqrt{p_i}e^{i\theta_i}\ket{i}$ in $\mathfrak{D}$. From Corollary \ref{rem:conv-pure2}, we can transform any state in $\mathfrak{D}$ to the ground states $\ket{0}$ using RPPO with $\sigma(\vec{p})=\sum_{i=0}^{d-1}p_i\op{i}$ and $\sigma(\vec{q})=\op{0}$. Now, since there is no entropy change in this process, the maximal work that can be extracted from $\ket{\psi}$ using RPPOs is given by
\begin{align*}
W_{max}&=\mathrm{Tr}[\hat{H}\left( \op{\psi} - \op{0}\right)]\\
&= \sum_{i=0}^{d-1} p_i E_i - E_0.
\end{align*}
The above equation tells us that $W_{max}$ is exactly equal to the ergotropy of the states in $\mathfrak{D}$. See Eq. \eqref{def:erg} for the definition of the ergotropy. This concludes the proof.
\end{proof}

The fact that we can bring  $\ket{\psi}$ to the ground state $\ket{0}$ and extract $W_{max}$ is trivial (it would be enough to use a unitary mapping $\ket{\psi}$ to $\ket{0}$). But what is non trivial here is that there exists a transformation bringing $\ket{\psi}$ to $\ket{0}$ and at the same time this transformation also maps any passive state that is virtually cooler than $\sigma(\vec{p})$ to $\op{0}$. This follows from Corollary \ref{rem:conv-pure2}.

\subsection{Monotones based on Hoffman majorization}
\label{append:mont}
We complete this Section by constructing some functions of states which are monotone under RPPOs, based on the consequences of Hoffman majorization. 
We have shown that, on the set of states $\mathfrak{D}$, the Hoffman majorization provides necessary and sufficient condition for state transformations under RPPOs. This can be exploited to construct a family of monotones that can only decrease under RPPOs. Let us recall a result on Hoffman majorization, which is stated as following theorem.

\begin{theorem}\cite{Marshall2011}
\label{theroem:mono}
Let $\msc{A}$ be the set of all $d\times d$ real symmetric matrices such that for all $A=(a_{ij})\in \msc{A}$, $i=0,\cdots, d-1$  and $k=0,\cdots, d-2$, $\sum_{j=1}^{i} (a_{k,j}-a_{k+1,j}) \geq 0$. Then, $x A x^T\leq y A y^T$ for all $A\in \msc{A}$ if and only if $x\prec_h y$.
\end{theorem}
Using above theorem, we have following proposition.
\begin{proposition}
For a qudit system with Hamiltonian $\hat{H}=\sum_{i=0}^{d-1} E_i \op{i}$ with $E_0\leq\cdots\leq E_{d-1}$, the following functions of the state $\ket{\psi}=\sum_{i=0}^{d-1} \sqrt{p_i}e^{i\theta_i} \ket{i}\in\mathfrak{D}$  are monotones under RPPOs.

\noindent
(1) function $\mathcal{A}_\alpha = E_{0}^{-\alpha}-\sum_{i=0}^{d-1}p_i^2 E_i^{-\alpha}$ for $\alpha\in(0,\infty)$;

\noindent 
(2) function $\mathcal{B}_\alpha =e^{-\alpha E_0}- \sum_{i=0}^{d-1}p_i^2 e^{-\alpha E_i}$ for $\alpha\in(0,\infty)$.

\noindent
These functions can therefore be labeled as resource quantifiers, while resource is being activity or nonpassivity.
\end{proposition}
\begin{proof}
Let $\ket{\psi}=\sum_{i=0}^{d-1}\sqrt{p_i}e^{i\theta_i}\ket{i}$ and $\ket{\phi}=\sum_{i=0}^{d-1}\sqrt{q_i}e^{i\gamma_i}\ket{i}$. We know from Theorem \ref{th:qud-iff} that if $\Phi_p(\op{\psi})=\op{\phi}$, then $\vec{p}\prec_h \vec{q}$. Now, from Theorem \ref{theroem:mono}, we know that for all $A=(a_{ij})\in \msc{A}$, $i=0,\cdots, d-1$  and $k=0,\cdots, d-2$ such that $ \sum_{j=0}^{i} (a_{k,j}-a_{k+1,j}) \geq 0$, $\vec{\psi} A \vec{\psi}^T\leq \vec{\phi} A \vec{\phi}^T$ for all $A\in \msc{A}$. For the particular choices $A =\hat{H}^{-\alpha}$ and $A=e^{-\alpha \hat{H}}$, we have $\sum_{i=0}^{d-1} p_i^2 E_i^{-\alpha} \leq \sum_{i=0}^{d-1} q_i^2 E_i^{-\alpha}$ and $\sum_{i=0}^{d-1} p_i^2 e^{-\alpha E_i} \leq \sum_{i=0}^{d-1} q_i^2 e^{-\alpha E_i}$. Thus
\begin{align*}
\mathcal{A}_\alpha(\Phi_p(\op{\psi}))&=\mathcal{A}_\alpha(\op{\phi})\\
&=E_0^{-\alpha}-\sum_{i=0}^{d-1} q_i^2 E_i^{-\alpha}\\
&\leq E_0^{-\alpha} - \sum_{i=0}^{d-1} p_i^2 E_i^{-\alpha}\\
&=\mathcal{A}_\alpha(\op{\psi}).
\end{align*}
Similarly, $\mathcal{B}_\alpha(\Phi_p(\op{\psi}))\leq \mathcal{B}_\alpha(\op{\psi})$. This, shows that the functions $\mathcal{A}_\alpha$ and $\mathcal{B}_\alpha $ are monotones under RPPOs on the set $\mathfrak{D}$.
\end{proof}

\section{Conclusion and discussion}
\label{sec:conc}
In this work, we have introduced a partial order on the set of passive states which generalizes the natural ordering of thermal states in terms of temperature. This provides us with a way to differentiate between passive states based on their usefulness in a thermodynamical context and paves the way to the notion of ``virtually cooler" passive states, making a clear connection with the notion of ``virtual temperature'' \cite{Brunner2012}. The order relation of ``being virtually cooler than'' as we define it is closely related to the likelihood ratio order appearing in statistics \cite{Shaked2007} and can be interpreted as a \textit{relative} passivity condition: a state that is passive \textit{relative to} a given passive state is also ``virtually cooler than'' this other state. We show that this partial order is stronger than majorization, which is only a preorder, in the sense that the partial order relation of ``being virtually cooler than'' implies a majorization relation. In a thermodynamical context, we show that, if used as a refrigerator, a virtually cooler state (with respect to another state) can cool down an external qubit system to a further extent as compared with the cooling effected by this other state.

This leads us to analyze the class of quantum channels that preserve this partial order relation of ``being virtually cooler than''. More specifically, given two fixed passive states $\rho$ and $\sigma$, we ask when does a quantum channel map a virtually cooler state than $\rho$ to a virtually cooler cooler state than $\sigma$? We call such a channel a relative passivity-preserving channel or operation (RPPO) with respect to the two fixed passive states $\rho$ and $\sigma$. We show that RPPOs are necessarily incoherent operations, which cannot create quantum coherence in the energy eigenbasis. We also compare RPPOs with the set of passivity-preserving channels or operations (PPOs), which are those quantum channels that output a passive state if the input is a passive state. We show that PPOs and RPPOs are in general two inequivalent notions. However, in case we define RPPOs by fixing both $\rho$ and $\sigma$ to be the maximally mixed state, then the set of RPPOs is equivalent to the set of PPOs (this is because passive states can be defined as the states that are virtually cooler than the maximally mixed state).

We then turn to the question of what quantum states can be converted into what other quantum states by a strictly incoherent  RPPO (i.e., having the extra property that all its Kraus operators are incoherent). We show that for a special class $\mathfrak{D}$ of active pure states, the interconversion ability of a strictly incoherent RPPO is equivalent to a particular partial order relation which we call Hoffman majorization. In particular, a pure state $\ket{\psi}=\sum_i\sqrt{p_i}e^{i\theta_i}\ket{i}\in \mathfrak{D}$  can be transformed to another pure state $\ket{\phi}=\sum_i\sqrt{q_i}e^{i\gamma_i}\ket{i}\in \mathfrak{D}$ if and only if $\vec{p}\prec_h\vec{q}$, where both $\vec{p}$ and $\vec{q}$ are passive vectors and $\prec_h$ stands for Hoffman majorization. Here, the two passive states with respect to which RPPOs are defined are $\sigma(\vec{p})=\sum_i p_i \op{i}$ and $\sigma(\vec{q})=\sum_i q_i \op{i}$. The proof of this result is constructive and rather tedious. So, in order to make it more instructive, we further elaborate in Appendix \ref{append:exp} on the explicit construction of a strictly incoherent RPPO implementing the desired state transformation for a general qutrit case ($d$=3). 

Just like the notion of PPOs comes with a resource theoretical interpretation of the states that are {\it not} passive, it is natural from the definition of RPPOs to view  the property of {\it not} being virtually cooler than a given passive state as a distinct resource. Accordingly, we introduce two families of resource monotones based on Hoffman majorization, which are nonincreasing under RPPOs. It would be quite relevant to better understand the nature of this particular resource, that is, ``not being virtually cooler than'', in a thermodynamical scenario. Further, it is worth noting here that in the absence of a heat bath, Hoffman majorization is a natural alternative for thermo-majorization \cite{Michal2013, Brandao2013, Brandao2015b}, which is meaningful only in the presence of a heat bath.

Interestingly, things become easier for qubit state transformations ($d$=2) and we show that Hoffman majorization becomes a necessary and sufficient condition for state transformations under PPOs (it is not needed  to consider RPPOs any more).  We then characterize the general passivity-preserving operations, which are also incoherent operations, and provide explicit forms of the Kraus-operators that comprise passivity-preserving operations in the qubit case. However, the characterization of passivity-preserving operations in arbitrary dimension is left open for future research. It would be very interesting to see whether Hoffman majorization, which preserves the nonincreasing nature of vectors, plays some particular role in the interconversion of pure $d$-dimensional states under passivity-preserving operations as it does for qubits.

As a limiting case of passivity-preserving operations, we also introduce the class of operations that always map any state (passive or active) to a passive state, and denote such operations as activity-breaking operations (ABOs). Interestingly, we show that these operations admit a clean characterization in terms of measure-and-prepare channels.

Finally, let us stress that we expect the notion of virtually cooler passive states introduced here to play an important role in thermodynamical contexts, as illustrated in Appendix~\ref{append:thermo} with the simple example of cooling an external qubit system using energy-preserving swap operations and a passive state for the refrigerator. It would actually be of great value to uncover the full implications of the relation of ``being virtually cooler than'' in more general thermodynamical contexts. In particular, considering the task of work extraction from a single quantum system under some quantum channel, it will be very interesting to see what are the consequences of using a virtually cooler passive state as the machine state. Importantly, to answer this question, one will need a good definition of extractable work under quantum channels (this is briefly discussed in Appendix~\ref{append:extractablework}). Unlike the case where unitaries are used for work extraction, this cannot simply be equal to the energy change of the system. The reason behind this is the fact that quantum channels introduce noise and this noise should be carefully separated from the energy change in order to determine the useful work. This is very intriguing as there is no notion of temperature or thermal bath in our scenario with passive states, so the usual separation, which is obtained by subtracting the entropy (times the temperature) from the internal energy, is simply not viable. In other words, this calls for a new notion of free energy which would generalize the usual free energy $F=E-TS$ in situations where no thermal bath at a given temperature is considered. A tempting possibility would be to use the notion of virtual temperatures instead, but we leave this question open in the present work.

\medskip

\begin{acknowledgments}
US and NJC acknowledge support from the F.~R.~S.-FNRS Foundation under Project No.~T.$0224.18$.
\begin{wrapfigure}{l}{0.15\textwidth}
    \includegraphics[width=0.167\textwidth]{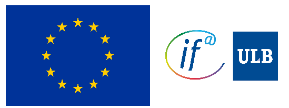}
\end{wrapfigure}
SD acknowledges Individual Fellowships at Universit\'{e} libre de Bruxelles; this project receives funding from the European Union's Horizon 2020 research and innovation programme under the Marie Sk\l odowska-Curie grant agreement No.~801505.
\end{acknowledgments}

\appendix

\section{Hoffman majorization}
\label{append:Hoff-Maj}
In this Appendix, we first review the proofs of Theorems~\ref{th:hoff1} and \ref{th:hoff2} at the core of the Hoffman majorization. The proofs in Sections \ref{append:A1} and \ref{append:A2} are inspired from Ref. \cite{Hoffman1969}. Then, we propose in Section \ref{append:sym-maj} an alternative condition for expressing Hoffman majorization 
that is based on what we call passive $t$-transforms. We show that Hoffman majorization can also be realized with asymmetric doubly-stochastic matrices (products of passive $t$-transforms) and explore in particular the cases of two and three dimensions. This analysis is also of independent mathematical interest.

\subsection{Proof of Theorem \ref{th:hoff1}}
\label{append:A1}
From the structure of the set $\mc{M}^{\mc{P}(d)}$ it is clear that any matrix $M^{\tau}\in \mc{M}^{\mc{P}(d)}$ satisfies the conditions $(a)$ to $(d)$ listed in Sec. \ref{sec:prelims-A}, which are necessary to be a Hoffman matrix. This follows by noting that for each partition $\tau=(\tau_1,\tau_2,\cdots \tau_k)$, $M^{\tau}$ is a $d\times d$ matrix, written as $M^{\tau}=\oplus_{t=1}^{k}M_{\tau_t}$, where $M_{\tau_t}=\frac{\mathfrak{I}_{|\tau_t |}}{|\tau_t |}$, $\mathfrak{I}_{|\tau_t |}$ is a $|\tau_t |\times |\tau_t |$ matrix of all ones and $|\tau_t |$ is the cardinality of part $\tau_t$. Now let $R\in\mc{R}(d)$ be a given Hoffman matrix. Our aim is to show that $R$ can be written as convex mixture of matrices $M^{\tau}\in \mc{M}^{\mc{P}(d)}$.

\bigskip
\noindent
{\bf Step 1 (Relationships between matrix elements of  $\mathbf{R}$)}. We will first establish some relations between matrix elements of $R$ and in particular, $R_{i,j}\geq 0$ for all $i,j$.
\begin{enumerate}
\item $R_{i,j}\geq R_{i,j+1}$ for $i\leq j$. We prove this by induction. This relation is trivially true for $i=0$ following from condition $(d)$. Let us assume that the relation is true for $i=k$, i.e., $R_{k,j}\geq R_{k,j+1}$ for $k+1\leq j$. Then from condition $(d)$ we have $R_{k+1,j}+R_{k,j+1}\geq R_{k,j}+ R_{k+1,j+1}$. Adding the two conditions, we get $R_{k+1,j}\geq R_{k+1,j+1}$. Thus the relation  is true for $i=k+1$. This completes the induction.
\item Similarly, $R_{i,j}\geq R_{i-1,j}$ for $1\leq i\leq j$. The proof is similar to the above relation.
\item Using the condition $R_{i,j}=R_{j,i}$, we have (i) $R_{i,j}\geq R_{i+1,j}$ for $i\geq j$, and (ii) $R_{i,j}\geq R_{i,j-1}$ for $1\leq j\leq i$.
\item Combining all the above three points and condition $(a)$ that $R_{0,d-1}\geq 0$, we conclude that $R_{i,j}\geq 0$ for all $i,j$.
\end{enumerate}

\bigskip
\noindent
{\bf Step 2 (Finding a partition based on step 1)}. We find a relevant partition recursively. Let us consider $i_1=\max j$ such that $R_{0,j}>0$. If $i_1< d-1$, then there exists at least one $j>i_1$ such that $R_{i_1+1,j}>R_{i_1,j}$. If this is not the case then from point 3 above we have $R_{i_1+1,j}\leq R_{i_1,j}$ for all $j$. Then from the condition that $1=\sum_{j}R_{i_1+1,j}\leq \sum_{j} R_{i_1,j}=1$ for all $j$, we have $R_{i_1+1,j}= R_{i_1,j}$. In particular, $R_{0,i_1+1}=R_{i_1+1,0}=R_{i_1,0}=R_{0,i_1}$. This contradicts the definition of $i_1$. If $i_1<d-1$, then let $i_2=\max j$ such that $R_{i_1+1,j}>R_{i_1,j}$. Now if $i_2<d-1$, then there exists at least one $j>i_2$ such that $R_{i_2+1,j}>R_{i_2,j}$. The proof of this fact is similar as the case for $i_1$. Carrying out this procedure results in an increasing sequence $0\leq i_1<\cdots i_{k}<d-1$. Now consider a partition $\tau=\left(\tau_1,\cdots, \tau_k\right)$, where
\begin{align*}
\tau_1&=(1,\cdots,i_1);\\
\tau_2&=(i_1+1,\cdots,i_2);\\
&\vdots\\
\tau_k&=(i_k,\cdots,d-1).
\end{align*}

\bigskip
\noindent
{\bf Step 3 (If $\mathbf{R}$ satisfies the conditions in its definition as equality then so does $\mathbf{M^{\tau}}$):--}If $R_{0,d-1}=0$ then $M^{\tau}_{0,d-1}=0$ as well as there is no sub-partition containing $0$ and $d-1$. Now for $i\leq j$ let
\begin{align*}
M^{\tau}_{i,j}+M^{\tau}_{i-1,j+1}> M^{\tau}_{i-1,j} + M^{\tau}_{i,j+1}.
\end{align*}
Then there exists an index $m<k$ and a partition $\tau_{m+1}=(i_m+1,\cdots i_{m+1})$ such that $i=i_m+1$ and $j=i_{m+1}$ and $i,j\in\tau_{m+1}$. This means that $R_{i,j}>R_{i-1,j}$ and $R_{i,j+1}=R_{i-1,j+1}$, which implies that $R_{i,j}+R_{i-1,j+1}> R_{i-1,j} + R_{i,j+1}$. By negation of this result, it is proved that if $R_{i,j}+R_{i-1,j+1}= R_{i-1,j} + R_{i,j+1}$ then $M^{\tau}_{i,j}+M^{\tau}_{i-1,j+1}= M^{\tau}_{i-1,j} + M^{\tau}_{i,j+1}$.

\bigskip
\noindent
{\bf Step 4 (Constructing $\mathbf{R}$ from $\mathbf{M^{\tau}}$):--}If $M^{\tau}$ constructed above is such that $R=M^{\tau}$ then we are done. Otherwise from the construction of $M^{\tau}$ above, there exists a small constant $\alpha>0$ such that $S=(R-\alpha M^{\tau})/(1-\alpha)\in \mc{R}(d)$. Let $\alpha_1=\max \alpha$ such that $S\in\mc{R}(d)$. Let $R_0=R$ and $M^{\tau^1}=M^{\tau}$, and $R_1=S$ then
\begin{align*}
R_0=(1-\alpha_1)R_1 +\alpha_1 M^{\tau^1}.
\end{align*}
Note that $R_1$ satisfies all the equalities in conditions $(a)$ to $(d)$ that $R_0$ satisfies. Moreover, there is at least one inequality that $R_0$ satisfies as strict inequality but $R_1$ satisfies as an equality. Now decomposing $R_1$ in similar way as $R_0=R$, we obtain a sequence of matrices $R_m$ and $M^{\tau^m}$, and constants $0<\alpha_1,\cdots,\alpha_m<1$ such that
\begin{align*}
&R_0=(1-\alpha_1)R_1 +\alpha_1 M^{\tau^1};\\
&R_1=(1-\alpha_2)R_2 +\alpha_2 M^{\tau^2};\\
&~~~\vdots\\
&R_{m-1}=(1-\alpha_m)R_m +\alpha_m M^{\tau^m}.
\end{align*}
Here, the number of inequalities in $R_m$ satisfied as equalities strictly increase with $m$. But since there are only finite such inequalities there will exist an $m$ such that $R_m\in \mc{M}^{\mc{P}(d)}$. Thus $R=R_0$ is the convex hull of $\mc{M}^{\mc{P}(d)}$ \cite{Hoffman1969}.

\subsection{Proof of Theorem \ref{th:hoff2}}
\label{append:A2}
Let $x,y\in\mathcal{S}(d)$ be two passive vectors. The `if' part of theorem is trivial, i.e., if $x=Ry$ and $R\in\mc{R}(d)$ then $x\prec_h y$ is trivially true. We will prove the `only if' part by induction. Again, the proof presented here relies on Ref. \cite{Hoffman1969}. For $d=1$, we have $x\prec_h y$ implies $x=y$ and $R$ is just $1$.

\medskip
\noindent
{\bf Case 1:--}Let $\sum_{i=0}^{k-1}x_i=\sum_{i=0}^{k-1}y_i$ for $k<d$. Also, let $x'\in\mc{S}(k)$ and $x''\in\mc{S}(d-k)$ be such that they coincide with first $k$ and last $d-k$ terms of $x$, respectively. Define $y'$ and $y''$ similarly. By construction $x'\prec_h y'$ and $x''\prec_h y''$, then by inductive assumption there exist Hoffman matrices $R'\in \mc{R}(k)$ and $R''\in \mc{R}(d-k)$ such that $x'=R'y$ and $x''=R''y''$. Then, $x=Ry$, where $R=R'\oplus R''$. It is easy to see that $R\in\mc{R}(d)$.

\medskip
\noindent
{\bf Case 2:--}Let $\sum_{i=0}^{k-1}x_i<\sum_{i=0}^{k-1}y_i$ for all $k<d$. Since $x\in\mathcal{S}(d)$, we have $\sum_{i=0}^{d-1}x_i\leq d x_1$ or $x_1\geq 1/d \sum_{i=0}^{d-1}x_i$. Now define $z=R_0y$, where $R_0=\mathfrak{I}_{d}/d$. Therefore we have 
\begin{align*}
x_1\geq \frac{1}{d} \sum_{i=0}^{d-1}x_i =\frac{1}{d} \sum_{i=0}^{d-1}y_i=z_0=\cdots=z_{d-1}.
\end{align*}
Now there exists a vector $w = \alpha z+(1-\alpha)y$, where $0 < \alpha \leq 1$ such that $\sum_{i=0}^{k-1}x_i\leq\sum_{i=0}^{k-1}w_i$ for $0\leq k\leq d-1$ and equality for $k=d$. Moreover, there is at least one $k\neq d$ such that $\sum_{i=0}^{k-1}x_i=\sum_{i=0}^{k-1}w_i$, then from case 1, there exists $R\in \mc{R}(d)$ such that $x=R w$, i.e.
\begin{align*}
x&=R(\alpha R_0y+(1-\alpha)y)\\
&=\left(\alpha R_0+(1-\alpha)R\right)y.
\end{align*}
Clearly $\alpha R_0+(1-\alpha)R\in\mc{R}(d)$, which completes the proof the theorem.

\subsection{Hoffman majorization and passive $t$-transforms}
\label{append:sym-maj}
It is known that the usual majorization relation $x\prec y$ between two $d$-dimensional probability vectors $x$ and $y$ is equivalent to the existence of a sequence $\mc{T}=(T_1,\cdots,T_l)$ of so-called $t$-transforms such that $x=\mc{T}y:=T_l\circ\cdots\circ T_1 y$ \cite{Marshall2011}. A $t$-transform $T$ is defined as a doubly stochastic matrix that acts nontrivially only on two components of the probability vector and the nontrivial $2\times 2$ block can be written as
\begin{align*}
\begin{pmatrix}
t & \bar{t}\\
\bar{t}& t
\end{pmatrix}
\end{align*}
with $0\leq t=1- \bar{t}\leq 1$.
Interestingly, the existence of a sequence of $t$-transforms is not specific to usual majorization, and it can be shown to exist for other types of majorization, such as $p$-majorization (see e.g. Chapter 14 of Ref. \cite{Marshall2011} for definition). Further, the existence of a sequence of $t$-transforms sometimes simplifies the mathematical analysis \cite{Zylka1985} and it is uniquely suited to certain physical processes, see e.g. Ref. \cite{Thon2004}. Therefore it is natural to ask whether there also exists a characterization of Hoffman majorization in terms of a sequence of some special $t$-transforms. Theorem~\ref{th:exist-t-for-hm} answers this question in the affirmative, building on the notion of \textit{passive t-transforms}. 
 We define a $t$-transform $T$ as passive if, for any $2\leq k \leq d$, it can be written as
\begin{align}
\label{eq:passive-t-trans}
T=\mathbb{I}_{k-2}\oplus \begin{pmatrix}t & \bar{t}\\ \bar{t}& t\end{pmatrix} \oplus \mathbb{I}_{d-k},
\end{align}
where $1/2\leq t=1-\bar{t}\leq 1$. Here and below we denote the $k\times k$ identity matrix by $\mathbb{I}_k$.  We stress that the two components of the probability vector that are acted upon in Eq. \eqref{eq:passive-t-trans} must be consecutive components, unlike for usual $t$-transforms.  We call $T$ a  passive $t$-transform because it maps the set of passive vectors into itself.
We do not use passive $t$-transforms in the context of the current paper, but we believe that this question is of independent interest for quantum thermodynamics dealing with passive states (note that passive $t$-transforms are also studied in economics under the name of `altruistic transfers', see e.g. Ref. \cite{Marshall2011}). 

\begin{theorem}
\label{th:exist-t-for-hm}
Let $x,y\in\mc{S}(d)$ be two passive vectors. The Hoffman majorization relation $x\prec_h y$ holds if and only if there exists a finite sequence $\mc{T}$ of passive $t$-transforms such that $x=\mc{T}y$.
\end{theorem}

\begin{proof}
We note first that for $x,y\in\mc{S}(d)$, the existence of a finite sequence $\mc{T}=(T_1,\cdots,T_l)$  of passive $t$-transforms such that $x=\mc{T}y$ is equivalent to the existence of a finite sequence $\mc{V}=\left(s^{(1)},\cdots,s^{(l)}\right)$ of vectors in $\mc{S}(d)$ such that $s^{(i)}$ and $s^{(i+1)}$ differ in two components only, as in Eq. \eqref{eq:passive-t-trans}, and $x=s^{(0)}\prec_h s^{(1)} \prec_h\cdots\prec_h s^{(l)}\prec_h s^{(l+1)}=y$ for some integer $l$. Then, the `if' part of the theorem is trivial, i.e, if there exists such a finite sequence $\mc{V}$, then $x\prec_h y$.

For the `only if' part we will prove the theorem by induction in the similar way as we proved Theorem \ref{th:hoff2}. For $d=2$, the theorem is obvious, i.e., for $x,y\in\mc{S}(2)$, then $x\prec_h y$ implies $x=Ry$, where $R$ is a $2\times 2$ Hoffman matrix which is also a $2\times 2$ passive $t$-transform
\begin{align*}
T= \begin{pmatrix}
t & \bar{t}\\
\bar{t}& t
\end{pmatrix}
\end{align*}
with $t\geq 1/2$ and $\bar{t}=1-t$. Thus the theorem is satisfied. Now two cases arise.

\medskip
\noindent
{\bf Case 1:--}Let $x\prec_h y$ and $\sum_{i=0}^{k-1}x_i=\sum_{i=0}^{k-1}y_i$ for $k<d$. Also, let $x'\in\mc{S}(k)$ and $x''\in\mc{S}(d-k)$ be such that they coincide with first $k$ and last $d-k$ terms of $x$, respectively. Define $y'$ and $y''$ similarly. By construction $x'\prec_h y'$ and $x''\prec_h y''$, then by inductive assumption there exist two sequences $\mc{T}_1=\{u^{(1)},\cdots,u^{(l)}\}$ and $\mc{T}_2=\{v^{(1)},\cdots,v^{(m)}\}$ of vectors in $\mc{S}(k)$ and $\mc{S}(d-k)$, respectively, such that $x'=u^{(0)}\prec_h u^{(1)} \prec_h\cdots\prec_h u^{(l)}\prec_h u^{(l+1)}=y'$ and $x''=v^{(0)}\prec_h v^{(1)} \prec_h\cdots\prec_h v^{(m)}\prec_h s^{(m+1)}=y''$. Then these two sequences can be composed using 
$x=(\mc{T}_1 y',\mc{T}_2y'')=( \mathbb{I}_k \oplus \mc{T}_2 )(\mc{T}_1\oplus\mathbb{I}_{d-k})(y',y'')$ and the theorem is satisfied for dimension $d$.

\medskip
\noindent
{\bf Case 2:--}Let $\sum_{i=0}^{k-1}x_i<\sum_{i=0}^{k-1}y_i$ for all $k<d$. Since, $x_0<y_0$, let us define $w=aTy+\bar{a}y$, where $T=\begin{pmatrix}t & \bar{t}\\ \bar{t} & t\end{pmatrix}\oplus \mathbb{I}_{d-2}$, $1/2\leq t=1-\bar{t}\leq 1$, and $0<a=1-\bar{a}\leq 1$, and let us choose $a$ such that $x_0=w_0$  (see e.g. Example \ref{ex:equal}). Now, we can apply case 1 to conclude that $x$ and $w$ satisfying $x\prec_h w$ satisfy the theorem. Note that $aT+\bar{a}\mathbb{I}$ itself a passive $t$-transform. Thus the theorem is also true for $w$ and $y$ satisfying $w\prec_h y$. Combining $x\prec_h w$ and $w\prec_h y$, we conclude that $x\prec_h y$ implies existence of a finite sequence $\mc{V}=\left(s^{(1)},\cdots,s^{(l)}\right)$ of vectors in $\mc{S}(d)$ such that $s^{(i)}$ and $s^{(i+1)}$ differ in two components only and $x=s^{(0)}\prec_h s^{(1)} \prec_h\cdots\prec_h s^{(l)}\prec_h s^{(l+1)}=y$. This concludes the proof of the theorem.
\end{proof}

\begin{example}[Construction of passive $t$-transforms]
\label{ex:equal}
Let us consider the two passive vectors $x=(0.6, 0.2, 0.1, 0.1)^T$
and $y=(0.65, 0.25, 0.05, 0.05)^T$ such that $\sum_{i=0}^{k-1} x_i < \sum_{i=0}^{k-1} y_i$, where $1\leq k < 4$, hence $x\prec_h y$. 
In the following, we show that we can always find a sequence $\mc{T}$ of passive $t$-transforms such that $x=\mc{T}y$. Let $S_0=\frac{1}{2}\begin{pmatrix}1 & 1\\ 1 &  1\end{pmatrix}$, $S=S_0\oplus \mathbb{I}_2$, $0< a \leq 1$, and define
\begin{align*}
w &:=a \, S y+(1-a)\, y=\begin{pmatrix}
-0.2 \, a + 0.65\\
0.2 \, a + 0.25\\
0.05\\
0.05
\end{pmatrix}.
\end{align*}
We can choose $a=1/4$ such that $x_0=w_0$ and in particular we have $w = (0.6, 0.3, 0.05, 0.05)^T$. 
Now, let $S'=1\oplus S_0\oplus 1$, $0<b\leq 1$, and define
\begin{align*}
w' &:=b \, S' w+(1-b)\, w=\begin{pmatrix}
0.6\\
-0.125 \, b + 0.3\\
0.125 \, b +0.05\\
0.05
\end{pmatrix}.
\end{align*}
Now we choose $b=0.8$ such that $w'_1=x_1$ and in particular, we have $w'=(0.6, 0.2, 0.15, 0.05 )^T$. Further, let $S''=\mathbb{I}_2\oplus S_0$, $0<c\leq 1$, and define 
\begin{align*}
w'' &:=c \, S'' w'+(1-c)\, w'=\begin{pmatrix}
0.6\\
0.2\\
-0.05 \, c +0.15\\
0.05 \, c +0.05
\end{pmatrix}.
\end{align*}
Now, we choose $c=1$ such that $w''_2=x_2$ and in particular, we have $w''=(0.6, 0.2, 0.1, 0.1 )^T=x$. Thus $x$ can be obtained from $y$ by applying the sequence
\begin{align*}
\mc{T}= S''(0.8 S'+0.2\mathbb{I}_4)(0.25 S + 0.75 \mathbb{I}_4).
\end{align*}
of passive $t$-transforms.
Note that the above decomposition is not unique as we could, for example, have constructed another sequence $\mc{T}$ by first applying $S_0=1\oplus \begin{pmatrix}0.5 & 0.5\\ 0.5 & 0.5\end{pmatrix}\oplus 1$ to $y$ and then following the above procedure starting from the highest weight in $y$.
\end{example}

Interestingly, we see that a sequence of passive $t$-transforms between vectors $x,y\in \mc{S}(d)$ satisfying $x\prec_h y$ results in a doubly-stochastic matrix that is most often asymmetric, despite the fact that Hoffman matrices are symmetric. This motivates us to define yet another equivalent condition for Hoffman majorization relying on what we call asymmetric Hoffman matrices, as exposed in the following Proposition.
\begin{proposition}
\label{prop:asym-hoff}
For two passive vectors $x,y\in\mc{S}(d)$, the condition $x\prec_h y$ holds if and only if there exists a doubly-stochastic matrix $D_{ah}$ such that $x=D_{ah}\, y$ and $D_{ah}\xi \in \mathcal{S}(d)$ for all $\xi\in \mathcal{S}(d)$. We call $D_{ah}$ an asymmetric Hoffman matrix (as it differs from the Hoffman matrix $R$) and define it as a $d\times d$ doubly stochastic matrix that satisfies a list of conditions as follows:
\begin{align*}
D_{ah}&=\begin{pmatrix}
a_{00} & \cdots & a_{0(d-1)}\\
\vdots& \ddots& \vdots\\
a_{(d-1)0} & \cdots & a_{(d-1)(d-1)}
\end{pmatrix},
\end{align*}
where $\sum_{i=0}^{d-1}a_{ij}= 1$,  $\forall j\in\{0,\cdots, d-1\}$, and $\sum_{j=0}^{d-1}a_{ij}=1$, $\forall i\in\{0,\cdots, d-1\}$, is an asymmetric Hoffman matrix provided it satisfies 
\begin{align*}
&\sum_{i=0}^j a_{0i} \geq \sum_{i=0}^j a_{1i}\geq \cdots \geq \sum_{i=0}^j a_{(d-1)i},~\forall~j \in \{ 0,\cdots, d-1\},
\end{align*}
where equality holds for $j=d-1$.
\end{proposition} 

\begin{proof}
If there exists $D_{ah}$ such that $x=D_{ah}\, y$, then $x\prec_h y$ is obvious since $D_{ah}$ is doubly-stochastic. The fact that $D_{ah}\, y$ is a passive vector when $x$ is a passive vector can be easily checked by considering the action of $D_{ah}$ on an extremal passive vector $\ket{e_j}$. We have
\begin{align*}
D_{ah}\ket{e_j}&=\frac{1}{j+1}\sum_{k,l=0}^{d-1} a_{kl}\ket{k}\bra{l}\sum_{i=0}^{j}\ket{i}\\
&=\frac{1}{j+1}\sum_{k,l=0}^{d-1} \sum_{i=0}^{j} a_{kl}\ket{k}\delta_{li}\\
&=\frac{1}{j+1}\sum_{k=0}^{d-1} \sum_{i=0}^{j} a_{ki}\ket{k}.
\end{align*}
For the output vector to be passive for all $j$, we need that $(D_{ah}\ket{e_j})_k\geq (D_{ah}\ket{e_j})_l$ for all $k\leq l$ which are the conditions in Proposition \ref{prop:asym-hoff}. This completes the proof of the ``if" part of the proposition.

Now, in order to prove the ``only if" part of the proposition, we assume that $x\prec_h y$. Then, from Theorem \ref{th:exist-t-for-hm}, we know there is a sequence $\mc{T}=(T_1,\cdots,T_l)$ of passive $t$-transforms such that $x=\mc{T}y$, and we must check that that it yields an asymmetric Hoffman matrix $D_{ah}$, satisfying the conditions in Proposition \ref{prop:asym-hoff}. Without loss of generality assume that the doubly stochastic matrix corresponding to $T_l$ is $D_l=\mathbb{I}_{d-2}\oplus \begin{pmatrix}t & \bar{t}\\ \bar{t} & t\end{pmatrix}$ with $1/2\leq t=1-\bar{t}\leq 1$ (see Example \ref{ex:equal}).  Now, let us assume inductively that the doubly stochastic matrix $D^{(l-1)}$ corresponding to $\mc{T}'=T_{l-1}\circ\cdots\circ T_1$ satisfies the conditions in Proposition \ref{prop:asym-hoff}. Note that $D'=D_l D^{(l-1)}$ is such that only last two rows of $D'$ are different from $D^{(l-1)}$. In fact, $D'_{d-2,i} = t D^{(l-1)}_{d-2,i}+\bar{t}D^{(l-1)}_{d-1,i}$ and $D'_{d-1,i} = \bar{t} D^{(l-1)}_{d-2,i}+ t D^{(l-1)}_{d-1,i}$ for all $i=0,\cdots, d-1$. It is easy to see that $D'$ also satisfies the conditions in Proposition \ref{prop:asym-hoff} and hence it is an asymmetric Hoffman matrix. Thus $x\prec_h y$ implies the existence of an asymmetric Hoffman matrix $D'$ such that $x=D'y$. 
\end{proof}

Let us illustrate Proposition \ref{prop:asym-hoff} in some simple cases. 
In the two-dimensional case, it is trivial to see that every passive $t-$transform satisfies condition $t \ge \bar{t}$ since 
with $t\geq 1/2$ and $\bar{t}=1-t$.
In the three-dimensional case, the composition of passive $t$-transforms on $\mc{S}(3)$ is analyzed in the following proposition.
\begin{proposition}
\label{prop:exist-t}
For three dimensional case, the ordered product of two $t$-transforms is passive, i.e., it maps $\mc{S}(d)$ into $\mc{S}(d)$, if and only if both the $t$-transforms are passive.
\end{proposition}

\begin{proof}
The two possible choices of ordered $t$-transforms are
\begin{align*}
T_1=\begin{pmatrix}
t & \bar{t} & 0\\
\bar{t} & t & 0\\
0 & 0 &1
\end{pmatrix};~ T_2=\begin{pmatrix}
1 & 0 & 0\\
0 & s & \bar{s}\\
0 & \bar{s} &s
\end{pmatrix},
\end{align*}
where $\bar{t}=1-t$ and $\bar{s}=1-s$. The possible composition of ordered $t-$transforms are $T_{12}=T_1T_2$ and $S_{21}=T_2T_1 $.

\bigskip
\noindent
{\it Case 1:} We have
\begin{align*}
T_{12}&=\begin{pmatrix}
t & s\bar{t} & \bar{s}\bar{t}\\
\bar{t} & st & \bar{s}t\\
0 & \bar{s} &s
\end{pmatrix}.
\end{align*}
Let us find the action of $T_{12}$ on extremal passive vectors.
\begin{align*}
T_{12}\ket{e_0}=\begin{pmatrix}
t\\
\bar{t}\\
0
\end{pmatrix}; T_{12}\ket{e_1}=\frac{1}{2}\begin{pmatrix}
t+s\bar{t}\\
\bar{t} + st\\
\bar{s}
\end{pmatrix}; T_{12}\ket{e_2}=\ket{e_2}.
\end{align*}
For $T_{12}\ket{e_0}$ to be passive, we need to have $t\geq \bar{t}$.
For $T_{12}\ket{e_1}$ to be passive, we need to have $(t-\bar{t})\bar{s}\geq 0$ and $s-\bar{s}t \geq 0$. The first inequality is always satisfied as $t\geq \bar{t}$ and the second inequality is always satisfied when $s\geq \bar{s}$.

\bigskip
\noindent
{\it Case 2:} We have
\begin{align*}
S_{21}=\begin{pmatrix}
t & \bar{t} & 0\\
s\bar{t} & st & \bar{s}\\
\bar{s}\bar{t} & \bar{s}t &s
\end{pmatrix}.
\end{align*}
Let us find the action of $S_{21}$ on extremal passive vectors.
\begin{align*}
S_{21}\ket{e_0}=\begin{pmatrix}
t\\
s\bar{t}\\
\bar{s}\bar{t}
\end{pmatrix};S_{21}\ket{e_1}=\frac{1}{2}\begin{pmatrix}
1\\
s\\
\bar{s}
\end{pmatrix};S_{21}\ket{e_2}=\ket{e_2}.
\end{align*}
For $S_{21}\ket{e_0}$ to be passive, we need to have $t-s\bar{t}\geq 0$ and $(s-\bar{s})\bar{t}\geq 0$, which implies $t\geq 1/2$ and $s\geq 1/2$.
For $S_{21}\ket{e_1}$ to be passive, we need to have $(1-s)\geq 0$ and $s-\bar{s} \geq 0$. The first inequality is always satisfied  and the second inequality is always satisfied when $s\geq \bar{s}$. This concludes the proof of the proposition.
\end{proof}

\begin{remark}
As noted in Example \ref{ex:equal}, the construction of a sequence of passive $t$-transforms from the highest to the lowest weight is not unique. This nonuniqueness is exemplified in Proposition \ref{prop:exist-t} since it may be possible to decompose a given asymmetric Hoffman matrix as $T_{12}$ or $S_{21}$. To be complete, let us note that it may also be the case that a given asymmetric Hoffman matrix does not admit a decomposition in terms of passive $t$-transforms. This is reminiscent of the fact that not every doubly stochastic matrix can be written as a product of $t$-transforms \cite{Marshall2011}.
\end{remark}

\section{PPO for qubits}
\label{append:qubits}
Since we have noted in Sec. \ref{sect:passive-states-PPO} that, for qubits, passivity-preserving operations are incoherent and moreover strictly incoherent, we can explicitly characterize the form of qubit passivity-preserving operations. 

\begin{remark}
When the input and output systems $S$ and $S'$ are two-dimensional (qubit) systems, all passivity-preserving operations  $\mc{N}_{S\to S'}$ can be expressed in terms of five (incoherent) Kraus operators with certain constraints.
\end{remark}

It is shown in Ref.~\cite {Streltsov2017} that any strictly incoherent operation admits a Kraus decomposition with at most $5$ incoherent Kraus operators. A canonical choice for such $5$ incoherent Kraus operators is given by $\{K_i\}_{i=1}^{5}$, where
\begin{align}
&K_1=\begin{pmatrix}a_1 & b_1 \\ 0 & 0\end{pmatrix},~
K_2=\begin{pmatrix}0 & 0 \\ a_2 & b_2\end{pmatrix},~
K_3=\begin{pmatrix}a_3 & 0 \\ 0 & b_3\end{pmatrix},\nonumber\\
&K_4=\begin{pmatrix}0 & b_4\\ a_4 & 0\end{pmatrix},~ K_5=\begin{pmatrix}a_5 & 0 \\ 0 & 0
\end{pmatrix}.
\end{align}
Here $a_i$ can be chosen as real and $b_i\in \mathbb{C}$. Further, $\sum_{i=1}^5 a_i^2=1=\sum_{i=1}^4|b_i|^2$ and $a_1b_1+a_2b_2=0$. It is to be noted that the above matrices are represented in the basis fixed to define incoherent operations. Here, let us fix  this reference basis to be the energy eigenbasis. For the incoherent Kraus operators to represent a passive operation, we must impose further restrictions and demand that $\Phi(\sigma_0^p)$ and $\Phi\left(\sigma_1^p\right)$ are both passive states, where $\sigma_0^p$ and $\sigma_1^p$ are extremal qubit passive states. These conditions are given by
\begin{align*}
a_1^2 + a_3^2+a_5^2  &\geq a_2^2+a_4^2;\\
a_1^2 + a_3^2+a_5^2 + |b_1|^2 + |b_4|^2 &\geq a_2^2+a_4^2 + |b_2|^2+|b_3|^2.
\end{align*}
The set of $5$ incoherent Kraus operators together with above constraints completely characterize qubit passivity-preserving operations.

Now, we may also exhibit another special feature of qubit passivity-preserving operations in terms of Stinespring dilation.

\begin{remark}
When $S$ is a qubit system, if a channel $\mc{N}$ is generated via 
\begin{align}
\label{eq:qub:pass}
\mc{N}(\rho_S)=\Tr_{E}\{\mathcal{U}_{SE}(\rho_S\otimes\sigma_{E})\},
\end{align}
where $\mathcal{U}_{SE}$ is an arbitrary energy-preserving unitary operation and $\sigma_{E}$ is a passive state of the environment, then $\mc{N}$ is a passivity-preserving channel. 
\end{remark}
\begin{proof}
The Hamiltonian of an arbitrary qubit system $S$ can be taken without loss of generality to be $\hat{H}_S=E\ket{1}\bra{1}$. An arbitrary energy preserving unitary $U_{SE}$ on the joint system $SE$ can be defined as follows
\begin{subequations}
\begin{align*}
&U_{SE}\ket{00} = \ket{00};\\
&U_{SE}\ket{01} = \alpha\ket{01}+\beta\ket{10};\\
&U_{SE}\ket{10} = -\beta^*\ket{01}+\alpha^*\ket{10};\\
&U_{SE}\ket{11} = \ket{11},
\end{align*}
\end{subequations}
where $|\alpha|^2 +|\beta|^2 =1$. Let $\sigma_E=q\ket{0}\!\bra{0}+\bar{q}\ket{1}\!\bra{1}$ be a passive state of the environment and $0\leq \bar{q}=1-q\leq q\leq 1$. Now, we have 
\begin{align*}
\mc{N}(\ket{0}\bra{0})&=\Tr_{E}\left\{U_{SE}(\ket{0}\!\bra{0} \otimes\sigma_{E})U_{SE}^\dagger\right\}\\
&=\Tr_{E}\left\{q\ket{00}\!\bra{00} \right. \\
&~~~~ \left. + \bar{q}(\alpha\ket{01}+\beta\ket{10})(\alpha^*\bra{01}+\beta^*\bra{10}) \right\}\\
&=\left( q + \bar{q} |\alpha|^2 \right)\ket{0}\!\bra{0}+ \bar{q} |\beta|^2\ket{1}\!\bra{1}.
\end{align*}
Now, using $\bar{q}\leq q$, we have $q + \bar{q} |\alpha|^2 -\bar{q} |\beta|^2 \geq q |\alpha|^2 + \bar{q} |\alpha|^2 \geq 0$. Thus $\mc{N}(\ket{0}\bra{0})$ is a passive state. Similarly
\begin{align*}
&\mc{N}\left(\frac{1}{2}\left(\ket{0}\!\bra{0} + \ket{1}\!\bra{1}\right)\right)\\
&=\frac{1}{2}\Tr_{E}\left\{q\ket{00}\!\bra{00} + \bar{q}\ket{11}\!\bra{11} \right. \\
&~~~~ \left. + \bar{q}(\alpha\ket{01}+\beta\ket{10})(\alpha^*\bra{01}+\beta^*\bra{10}) \right.\\
&~~~~ \left. + q(-\beta^*\ket{01}+\alpha^*\ket{10})(-\beta\bra{01}+\alpha\bra{10}) \right\}\\
&=\frac{1}{2}\left( q + \bar{q} |\alpha|^2 + q |\beta|^2 \right)\ket{0}\!\bra{0}\\
&~~~~ +\frac{1}{2} \left(\bar{q}+ \bar{q} |\beta|^2 + q |\alpha|^2 \right)\ket{1}\!\bra{1}.
\end{align*}
Again the above state is  a passive state because
\begin{align*}
&\left( q + \bar{q} |\alpha|^2 + q |\beta|^2 \right)- \left(\bar{q}+ \bar{q} |\beta|^2 + q |\alpha|^2 \right)\\
&=(q-\bar{q}) (1- |\alpha|^2 + |\beta|^2 )\\
&=2(q-\bar{q}) |\beta|^2 \geq 0.
\end{align*}
Thus $\mc{N}$ is a passivity-preserving channel.
\end{proof}

\medskip
Note, however, that for more than two-dimensional systems, such a channel based on a energy-preserving unitary operation and passive environment ceases to be necessarily passivity-preserving. In the following we construct such a channel for a qutrit system. Let us consider the system  Hamiltonian to be $\hat{H}_S=\ket{1}\bra{1}+2\op{2}{2}$. Also, consider an energy preserving unitary $U_{SE}$ such that $U_{SE}\ket{ij}=\ket{ij}$ for all $i,j=0,1,2$ except $(i=0, j=2)$ and $(i=2, j=0)$ for which $U_{SE}\ket{ij}=\ket{ji}$. Then, for any passive state $\sigma_E=\sum_{i=0}^2q_i\op{i}$ ($q_0\geq q_1\geq q_2$ and $\sum_{i=0}^2q_i=1$) of environment, we have
\begin{align*}
\mc{N}\left(\ket{0}\!\bra{0} +\ket{1}\!\bra{1}\right)&=\Tr_{E}\left\{U_{SE}(\ket{0}\!\bra{0} \otimes\sigma_{E})U_{SE}^\dagger\right\}\\
&=\Tr_{E}\left\{\sum_{i=0}^1q_i\ket{0i}\!\bra{0i} + q_2\ket{20}\!\bra{20} \right\}\\
&=(q_0+q_1)\ket{0}\!\bra{0} + q_2 \op{2}{2}.
\end{align*}
The above is clearly not a passive state and hence $\mathcal{N}$ is not a passivity-preserving channel.

Note that some particular classes of passivity-preserving channels have been introduced in Refs. \cite{Jabbour2016, Jabbour2019} in the context of continuous variable quantum systems, which has the form as Eq. \eqref{eq:qub:pass}, where $\mathcal{U}_{SE\to SE}$ could be either energy preserving or energy difference preserving Gaussian unitary operation.

\section{RPPO for qutrits}
\label{append:exp}
Here we show an example of qutrit case to illustrate the construction of a desired RPPO for the pure state transformations on $\mathfrak{D}$.
\begin{example}
A qutrit state $\ket{\psi} \in \mathfrak{D}$ can be transformed to another state $\ket{\phi}\in \mathfrak{D}$ under RPPOs if $\vec{p}\prec_h \vec{q}$, where $\ket{\psi}=\sum_{i=0}^{2}\sqrt{p_i}e^{i\theta_i}\ket{i}$, $\ket{\phi}=\sum_{i=0}^{2}\sqrt{q_i}e^{i\gamma_i}\ket{i}$, $\vec{p}=\{p_0,p_1,p_2\}$, and $\vec{q}=\{q_0,q_1,q_2\}$.
\end{example}
\begin{proof}
Again, as argued in the main text, the phase factors in the states $\ket{\psi}$ and $\ket{\phi}$ can be ignored. Since $\vec{p}\prec_h \vec{q}$, from Theorems \ref{th:hoff1} and \ref{th:hoff2}, we have
$\vec{p}=\sum_{\tau\in\mathcal{P}(3)}\alpha_\tau M^{\tau}\vec{q} = \sum_{i=1}^4\vec{r_i}$, where $\vec{r_i}=\alpha_{\tau^{(i)}} M^{\tau^{(i)}}\vec{q}$. We have provided $M^{\tau}$ explicitly in preliminary section (Sec. \ref{sec:prelims}). We have
\begin{align*}
\vec{r}_1=\alpha_{\tau^{(1)}}\begin{pmatrix}
q_0\\
q_1\\
q_2
\end{pmatrix}.
\end{align*}
Now, we define following Kraus operator
\begin{align*}
&K_0=\sqrt{\alpha_{\tau^{(1)}}}\begin{pmatrix}
\sqrt{\frac{q_0}{p_0}} & 0 &0\\
0 & \sqrt{\frac{q_1}{p_1}} & 0\\
0 & 0 & \sqrt{\frac{q_2}{p_2}}
\end{pmatrix}
\end{align*}
such that $K_0(\op{\psi}) K_0^\dagger= \alpha_{\tau^{(1)}} \op{\phi}$.
Further,
\begin{align*}
K_0^\dagger K_0&=  \sum_{i=0}^2\frac{(\vec{r_1})_i}{p_i}\op{i}.
\end{align*}
We also have
\begin{align*}
&K_0 \op{i} K_0^\dagger= \alpha_{\tau^{(1)}}\frac{q_i}{p_i} \op{i}.
\end{align*}
Let us consider another vector $\vec{r_2}$ such that
\begin{align*}
\vec{r}_2= \alpha_{\tau^{(2)}} \begin{pmatrix}
(q_0+q_1)/2\\
(q_0+q_1)/2\\
q_2
\end{pmatrix}.
\end{align*}
Now, let us introduce following Kraus operators
\begin{align*}
&K_1= \sqrt{\alpha_{\tau^{(2)}} }\begin{pmatrix}
\sqrt{\frac{q_0}{2p_0}} & 0 &0\\
0 & \sqrt{\frac{q_1}{2p_1}} & 0\\
0 & 0 & \sqrt{\frac{q_2}{2p_2}}
\end{pmatrix};\\
&K_2=\sqrt{\alpha_{\tau^{(2)}} }\begin{pmatrix}
0 & \sqrt{\frac{q_0}{2p_1}} &0\\
\sqrt{\frac{q_1}{2p_0}} & 0 & 0\\
0 & 0 & \sqrt{\frac{q_2}{2p_2}}
\end{pmatrix}.
\end{align*}
Then, we have 
\begin{align*}
&\sum_{i=1}^2 K_i\op{\psi} K_i^\dagger= \alpha_{\tau^{(2)}} \op{\phi}.
\end{align*}
Further,
\begin{align*}
\sum_{i=1}^2 K_i^\dagger K_i= \sum_{i=0}^2 \frac{(\vec{r_2})_i}{p_i} \op{i}.
\end{align*}
We also have
\begin{align*}
&\sum_{i=1}^2 K_i \op{0} K_i^\dagger =\frac{\alpha_{\tau^{(2)}}}{2 p_0} \sum_{i=0}^1 q_{i} \op{i};\\
&\sum_{i=1}^2 K_i \op{1} K_i^\dagger=\frac{\alpha_{\tau^{(2)}}}{2 p_1} \sum_{i=0}^1 q_{i} \op{i}; \\
&\sum_{i=1}^2 K_i \op{2} K_i^\dagger = \frac{\alpha_{\tau^{(2)}}}{ p_2} q_2 \op{2}.
\end{align*}
Let us consider another vector $\vec{r}_3$ such that
\begin{align*}
\vec{r}_3=\alpha_{\tau^{(3)}}\begin{pmatrix}
q_0\\
(q_1+q_2)/2\\
(q_1+q_2)/2
\end{pmatrix}.
\end{align*}
Let us define following Kraus operators
\begin{align*}
&K_3= \sqrt{\alpha_{\tau^{(3)}} }\begin{pmatrix}
\sqrt{\frac{q_0}{2p_0}} & 0 &0\\
0 & \sqrt{\frac{q_1}{2p_1}} & 0\\
0 & 0 & \sqrt{\frac{q_2}{2p_2}}
\end{pmatrix};\\
&K_4=\sqrt{\alpha_{\tau^{(3)}} }\begin{pmatrix}
\sqrt{\frac{q_0}{2p_0}}  & 0 &0\\
0 & 0 & \sqrt{\frac{q_1}{2p_2}}\\
0 & \sqrt{\frac{q_2}{2p_1}} & 0
\end{pmatrix}
\end{align*}
such that we have
\begin{align*}
&\sum_{i=3}^4 K_i\op{\psi} K_i^\dagger= \alpha_{\tau^{(3)}} \op{\phi}.
\end{align*}
Further,
\begin{align*}
&\sum_{i=3}^4 K_i^\dagger K_i
= \sum_{i=0}^2 \frac{(\vec{r_3})_i}{p_i} \op{i}.
\end{align*}
We also have
\begin{align*}
&\sum_{i=3}^4 K_i \op{0} K_i^\dagger= \frac{\alpha_{\tau^{(3)}}}{ p_0} q_0 \op{0};\\
&\sum_{i=3}^4 K_i \op{1} K_i^\dagger =\frac{\alpha_{\tau^{(3)}}}{2 p_1} \sum_{i=1}^2 q_i\op{i};\\
&\sum_{i=3}^4 K_i \op{2} K_i^\dagger =\frac{\alpha_{\tau^{(3)}}}{2 p_2}\sum_{i=1}^2 q_i\op{i}.
\end{align*}
Let us consider another vector $\vec{r}_4$ such that
\begin{align*}
\vec{r}_4=\frac{\alpha_{\tau^{(4)}}}{3}\begin{pmatrix}
1\\
1\\
1
\end{pmatrix}.
\end{align*}
Let us define a map $\Phi_4$ with Kraus operators
\begin{align*}
&K_5= \sqrt{\alpha_{\tau^{(4)}} }\begin{pmatrix}
\sqrt{\frac{q_0}{3p_0}} & 0 &0\\
0 & \sqrt{\frac{q_1}{3p_1}} & 0\\
0 & 0 & \sqrt{\frac{q_2}{3p_2}}
\end{pmatrix};\\
&K_6=\sqrt{\alpha_{\tau^{(4)}} }\begin{pmatrix}
0 & \sqrt{\frac{q_0}{3p_1}}  &0\\
0 & 0 & \sqrt{\frac{q_1}{3p_2}}\\
\sqrt{\frac{q_2}{3p_0}}  & 0 & 0
\end{pmatrix};\\
&K_7=\sqrt{\alpha_{\tau^{(4)}} }\begin{pmatrix}
 0 & 0 &\sqrt{\frac{q_0}{3p_2}}\\
\sqrt{\frac{q_1}{3p_0}} & 0 & 0\\
0 & \sqrt{\frac{q_2}{3p_1}} & 0
\end{pmatrix}.
\end{align*}

\begin{align*}
&K_5=\sqrt{\alpha_{\tau^{(4)}} } \sum_{i=0}^2 \sqrt{\frac{q_{\pi_5(i)}}{3p_i}}\ket{{\pi_5(i)}}\!\bra{i};\\
&K_6=\sqrt{\alpha_{\tau^{(4)}} } \sum_{i=0}^2 \sqrt{\frac{q_{\pi_6(i)}}{3p_i}}\ket{{\pi_6(i)}}\!\bra{i};\\
&K_7=\sqrt{\alpha_{\tau^{(4)}} } \sum_{i=0}^2 \sqrt{\frac{q_{\pi_7(i)}}{3p_i}}\ket{{\pi_7(i)}}\!\bra{i},
\end{align*}
where $\pi_5(0)=0$, $\pi_5(1)=1$, $\pi_5(2)=2$; $\pi_6(0)=2$, $\pi_6(1)=0$, $\pi_6(2)=1$; $\pi_7(0)=1$, $\pi_7(1)=2$, $\pi_7(2)=0$. We have,
\begin{align*}
&\sum_{i=5}^7 K_i \op{\psi}K_i^\dagger = \alpha_{\tau^{(4)}} \op{\phi}.
\end{align*}
Further,
\begin{align*}
\sum_{i=5}^7 K_i^\dagger K_i&= \sum_{i=0}^2 \frac{(\vec{r_4})_i}{p_i} \op{i}.
\end{align*}
We also have
\begin{align*}
&\sum_{i=5}^7 K_i \op{0} K_i^\dagger= \frac{\alpha_{\tau^{(4)}}}{3 p_0}  \sum_{i=0}^2  q_i \op{i};\\
&\sum_{i=5}^7 K_i \op{1} K_i^\dagger= \frac{\alpha_{\tau^{(4)}}}{3 p_1}  \sum_{i=0}^2  q_i \op{i}; \\
&\sum_{i=5}^7 K_i \op{2} K_i^\dagger= \frac{\alpha_{\tau^{(4)}}}{3 p_2}  \sum_{i=0}^2  q_i \op{i}.
 \end{align*}
Thus we see that $\sum_{i=0}^7 K_i \op{\psi}K_i^\dagger = \op{\phi}$. Moreover,
\begin{align*}
\sum_{i=0}^7 K_i^\dagger K_i&= \begin{pmatrix}
\frac{\left(\sum_{i=1}^4\vec{r_i}\right)_0}{p_0} & 0 &0\\
0 & \frac{\left(\sum_{i=1}^4\vec{r_i}\right)_1}{p_1}  & 0\\
0 & 0 &\frac{\left(\sum_{i=1}^4\vec{r_i}\right)_2}{p_2}
\end{pmatrix}\\
&=\mathbb{I}_3.
\end{align*}
Furthermore,
\begin{align*}
\Phi_p(\op{i}) = \sum_{a=0}^7 K_a \op{i} K_a^\dagger&=\frac{1}{p_i} \sum_{j=0}^{2} q_j  R_{i,j} \op{j}.
\end{align*}
Since the operation $\Phi_p$ has a similar structure as the one appearing in the Theorem \ref{th:qud-iff}, we conclude that $\Phi_p$ is a RPPO. This finishes the proof of the example.
\end{proof}

\section{Thermodynamical interpretation of the relative passivity relation}
\label{append:thermo}
In this Appendix, we provide a thermodynamical interpretation of the relative passivity relation (i.e., the fact that a passive state is virtually cooler than another passive state) in the context of a quantum refrigerator. Consider a passive state $\rho(\vec{r})$. With above definitions, we can consider a virtual qubit corresponding to levels $(0, d-1)$ and the corresponding virtual temperature $\beta_v$, defined via $r_0/r_{d-1}=e^{\beta_v E_v}$ with $E_v=\Delta E_{d-1,0}$, is given by
\begin{align}
\beta_v=\frac{1}{E_v} \sum_{i=0}^{d-2}\beta_{i,i+1}\Delta E_{i+1,i}.
\end{align}
The virtual qubit can be characterized by two parameters, namely, the normalization $P_v:=r_0+r_{d-1}$ and the bias $B_v:=(r_0-r_{d-1})/P_v=\tanh(\beta_vE_v/2 )$. Thus $B_v\rightarrow 1$ implies $\beta_v\rightarrow \infty$ or vanishing of the virtual temperature. Now consider an external qubit, with a Hamiltonian $H_{\ex}=f_0\op{u_0}+f_1\op{u_1}$, in a state $\rho_{\ex}=u_0\op{u_0}+u_1\op{u_1}$, where $(u_0, u_1)\geq 0$, $u_0+u_1=1$, and the energy gap is equal to $0\leq f_1-f_0=E_v$. The bias $B_{\ex}$ of the external qubit is given by $u_0-u_1$ and normalization is one. If we apply an energy conserving swap on the external qubit and the virtual qubit, then the final bias of the  external qubit becomes $B_\mathrm{fin}=P_v B_v + (1-P_v)B_{\ex}$ (see Ref. \cite{Silva2016} for an easy proof). The desired final bias is $B_\mathrm{fin}\rightarrow 1$ in case we want to cool the external system. 

A similar analysis can be done starting from another passive state $\sigma(\vec{p})$ and in this case, the final bias of the external qubit becomes $\bar{B}_\mathrm{fin}=\bar{P}_v \bar{B}_v + (1-\bar{P}_v)B_{\ex}$, where $\bar{P}_v=p_0+p_{d-1}$ and $\bar{B}_v=(p_0-p_{d-1})/\bar{P}_v$ are normalization and bias of the new virtual qubit, respectively. In the following we will prove that if $\rho(\vec{r})$ is virtually cooler than $\sigma(\vec{p})$, then $\beta_v \geq \beta'_v$. It also implies that $\rho(\vec{r})$ can cool the external qubit further than $\sigma(\vec{p})$.

\begin{proposition}
Consider a protocol where a virtual qubit of either $\rho(\vec{r})$ or $\sigma(\vec{p})$ is used to cool the external qubit. If $\rho(\vec{r})\succ_{vc}\sigma(\vec{p})$, then 
 $B_\mathrm{fin}\geq \bar{B}_\mathrm{fin}$, that is, we attain a larger bias with the virtually cooler state $\rho(\vec{r})$.  Denoting as $F$ and $\bar{F}$ the energy of the final state of the external qubit when using state $\rho(\vec{r})$ and $\sigma(\vec{p})$, respectively, then $F\leq \bar{F}$. 
\end{proposition}
\begin{proof}
To show the first part, consider
\begin{align*}
B_\mathrm{fin} - \bar{B}_\mathrm{fin}&=(r_0-r_{d-1} - p_0 +p_{d-1})\\
 &~~~~+ (-r_0 - r_{d-1} + p_0 + p_{d-1}) B_{\ex}\\
 &= 2(r_0-p_0)u_1 + 2(p_{d-1}-r_{d-1})u_0\\
 &\geq 0,
\end{align*}
where in the last line we have used $(u_0,u_1)\geq 0$ and the fact that $\vec{r}\succ_h \vec{p}$, which implies $r_0\geq p_0$ and $r_{d-1}\leq p_{d-1}$. This completes the first part of the proposition.

For the second part, notice that the final state of the external qubit in the two cases when the virtual qubit belongs to state $\rho(\vec{r})$ and $\sigma(\vec{p})$ is given by, say $\tau$ and $\bar{\tau}$, respectively, where
\begin{align*}
&\tau=\frac{1+B_\mathrm{fin}}{2}\op{u_0}+ \frac{1-B_\mathrm{fin}}{2}\op{u_1}\\
&\bar{\tau}=\frac{1+\bar{B}_\mathrm{fin}}{2}\op{u_0}+ \frac{1-\bar{B}_\mathrm{fin}}{2}\op{u_1}
\end{align*}
so that
\begin{align*}
F - \bar{F}&=\frac{(B_\mathrm{fin} - \bar{B}_\mathrm{fin})}{2}(f_0 -f_1)\leq 0,
\end{align*}
where we used the first part of the proposition and the condition that $f_0-f_1=-E_v\leq 0$. This completes the proof of the proposition.
\end{proof}

\section{Extractable work under a quantum channel}
\label{append:extractablework}

In quantum processes where no entropy change of states is involved, the extractable work is defined as the decrease in the energy of the system. However, if we consider maximal work extraction from an arbitrary state using RPPOs, we need to define the notion of extractable work for quantum channels. Traditionally, to describe thermodynamic work one needs some notion of temperature. For example, in the resource theory of thermodynamics \cite{Michal2013, Brandao2013, Brandao2015b}, one considers an external bath at some fixed temperature and then let the system interact with it via energy conserving unitaries. In such a situation the maximal extractable work can be expressed in terms of the difference between min-free energies of the system and the bath. The min-free energy $F_{\min}$ of a system in state $\rho$ and in the presence of bath at inverse temperature $\beta$ is given by
\begin{align}
F_{\min}(\rho):=F_{\beta} + D_{\min}(\rho||\rho_{\beta}),
\end{align}
where $\rho_{\beta}$ is the thermal state of the system at inverse temperature $\beta$, $F_\beta=E(\rho_{\beta})-\beta^{-1}S(\rho_{\beta})$ is usual free energy of the thermal state $\rho_{\beta}$, and $D_{\min}$ is the min-relative entropy. The min-relative entropy is defined as
\begin{align}
D_{\min}(\rho||\sigma):=-\ln \mathrm{Tr}\left[ \Pi_\rho \sigma\right],
\end{align}
where  $\Pi_\rho$ is the projector on the support of $\rho$.

Now, we know that for processes that do not allow for entropy change, the extractable work can be defined even in the absence of the notion of temperature or heat bath. Therefore it seems plausible to look for a notion of extractable work for the CPTP maps that do allow for entropy change, even in the absence of a heat bath. One possible way is to consider the following state transformation. Let $\rho_S$ be the state of the system of interest and $\sigma_A(\vec{p})$ be a fixed passive state of some ancilla system. Let us consider yet another system $S'$, namely a work storage device in a state $\ket{e_0}_{S'}$, where $\ket{e_0}$ is the ground state of the Hamiltonian of the work storage device with ground state energy being $e_0$. Now, consider all the transformations of the form
\begin{align}
&\Lambda(\rho_S\otimes \op{e_0}_{S'})\nonumber\\
&:=\mathrm{Tr}_A\left[U_{SS'A}\left(\rho_S\otimes \op{e_0}_{S'}\otimes \sigma_A(\vec{p})\right)U_{SS'A}^\dagger\right],
\end{align}
where $U_{SS'A}$ is an energy preserving unitary on the total system $SS'A$. Then, the extractable work from $\rho_S$ under the process $\Lambda$ can be defined as the change in energy $e_k-e_0$ in the transformation $\rho_S\otimes \op{e_0}_{S'} \xrightarrow{\Lambda}{} \eta_S\otimes \op{e_k}_{S'}$. Then, the maximization should be performed over all energy preserving unitaries $U_{SS'A}$ and final state $\eta_S$ of the system. Although this procedure is very general, it defines extractable work for CPTP maps that admit a dilation over energy preserving unitaries. The analysis of such a notion of extractable work is beyond the scope of the current article. Nevertheless, it raises interesting questions, e.g., can it be shown that all RPPOs admit such a dilation? As a starting point, we provide here a qubit example of a RPPO that admits such a dilation and discuss work extraction under this operation.

\begin{example}
Following the notations in the main text, let us fix two qubit passive states $\sigma(\vec{p})$ and $\sigma(\vec{q})$ such that $p_0=aq_0+\bar{a}q_1$, where $1/2\geq \bar{a}=1-a$. Let us consider a quantum channel $\Phi$ with Kraus operators
\begin{align*}
&L_1=\sqrt{a\frac{q_0}{p_0}}\op{0}{0}+\sqrt{a\frac{q_1}{p_1}}\op{1}{1};\\
&L_2=\sqrt{\bar{a}\frac{q_0}{p_1}}\op{0}{1};~L_3=\sqrt{\bar{a}\frac{q_1}{p_0}}\op{1}{0}.
\end{align*}
It is easy to check that $\sum_{i=1}^3L_i^\dagger L_i=\mathbb{I}$. Moreover, for a passive state $\sigma(\vec{r})$ that is virtually cooler than $\sigma(\vec{p})$, i.e., $\sigma(\vec{r})\succ_{vc}\sigma(\vec{p})$, we have
\begin{align*}
\sigma(\vec{s})&:=\Phi(\sigma(\vec{r}))\\
&= q_0\left(a \frac{r_0}{p_0} + \bar{a} \frac{r_1}{p_1}\right)\op{0}{0}+q_1\left(\bar{a} \frac{r_0}{p_0} + a \frac{r_1}{p_1}\right)\op{1}{1} ,
\end{align*}
or equivalently
\begin{align*}
\begin{pmatrix}
s_0/q_0\\
s_1/q_1
\end{pmatrix}=\begin{pmatrix}
a & \bar{a}\\
\bar{a} & a
\end{pmatrix}\begin{pmatrix}
r_0/p_0\\
r_1/p_1
\end{pmatrix}.
\end{align*}
Thus, $\Phi\in\mathfrak{L}_{p,q}$. The channel $\Phi$ can be realized using energy-preserving unitary on the system and a qutrit ancilla as follows
\begin{align*}
\Phi(\rho_S)=\mathrm{Tr}_{A}\left[U_{SA}\left(\rho_S\otimes \op{1}{1}_A\right) U^{\dagger}_{SA}\right],
\end{align*}
where $U_{SA}$ is an energy-preserving unitary, $H_S=\op{1}{1}$ and $H_A=\op{1}{1}+2\op{2}{2}$ are the Hamiltonians of the systems $S$ and ancilla $A$, respectively. In the following, we will omit the system and ancilla labels for brevity. The unitary $U_{SA}$ is defined as
\begin{align*}
&U_{SA}\ket{00}=\ket{00};~ U_{SA}\ket{12}=\ket{12};\\
&U_{SA}\ket{01}=\sqrt{a\frac{q_0}{p_0}}\ket{01}+\sqrt{\bar{a}\frac{q_1}{p_0}}\ket{10};\\
&U_{SA}\ket{10}=-\sqrt{\bar{a}\frac{q_1}{p_0}}\ket{01}+\sqrt{a\frac{q_0}{p_0}}\ket{10};\\
&U_{SA}\ket{02}=\sqrt{a\frac{q_1}{p_1}}\ket{02}-\sqrt{\bar{a}\frac{q_0}{p_1}}\ket{11};\\
&U_{SA}\ket{11}=\sqrt{\bar{a}\frac{q_0}{p_1}}\ket{02}+\sqrt{a\frac{q_1}{p_1}}\ket{11}.
\end{align*}
Let $\ket{\psi}_S=\alpha\ket{0} + \beta\ket{1}$, where $|\alpha|^2+|\beta|^2=1$. Then,
\begin{align*}
U_{SA}\ket{\psi,1}_{SA}=&\alpha\left(\sqrt{a\frac{q_0}{p_0}}\ket{01}+\sqrt{\bar{a}\frac{q_1}{p_0}}\ket{10}\right)\\
&+\beta\left(\sqrt{\bar{a}\frac{q_0}{p_1}}\ket{02}+\sqrt{a\frac{q_1}{p_1}}\ket{11}\right).
\end{align*}
Now we perform energy measurement on the ancilla, and the probabilities that the ancilla is in state $\ket{0}$, $\ket{1}$, and $\ket{2}$ are given by $|\alpha|^2 \bar{a}q_1/p_0$, $|\alpha|^2aq_0/p_0+|\beta|^2aq_1/p_1$, and $|\beta|^2\bar{a}q_0/p_1$, respectively. Then on an average energy increase of ancilla is given by
\begin{align*}
\Delta E_A&=\frac{|\alpha|^2aq_0}{p_0}+\frac{|\beta|^2aq_1}{p_1}+2\frac{|\beta|^2\bar{a}q_0}{p_1} -1\\
&=\frac{|\alpha|^2aq_0}{p_0}+\frac{|\beta|^2\bar{a}q_0}{p_1} -|\alpha|^2\\
&=-\frac{|\alpha|^2\bar{a}q_1}{p_0} +\frac{|\beta|^2\bar{a}q_0}{p_1}\\
&=\frac{\bar{a}q_0}{p_1}-|\alpha|^2\bar{a}\left(\frac{q_1}{p_0} +\frac{q_0}{p_1}\right).
\end{align*}
$\Delta E_A$ can be understood as the average work extracted by the RPPO $\Phi\in\mathfrak{L}_{p,q}$. In particular, if $\alpha=0$, then the average extracted work from $\op{1}{1}$ by $\Phi$ is given by $\frac{\bar{a}q_0}{p_1}\geq 0$.
\end{example}

\bibliography{passivity}
\end{document}